\documentclass[journal]{IEEEtran}
\usepackage{multicol}
\usepackage{amsmath,epsfig,comment}
\usepackage{amsthm}
\usepackage{tcolorbox}
\usepackage{amssymb}\usepackage{multirow}
\usepackage{mathrsfs}
\usepackage{amsmath}
\usepackage{arydshln}
\usepackage{multirow}
\usepackage{arydshln}

\usepackage{cases} 
\usepackage{amssymb,amsmath,cite}
\usepackage{epsfig}
\usepackage{color}
\usepackage{bm}
\usepackage{graphicx,subfigure}
\usepackage{algorithm}
\usepackage{algpseudocode}
\usepackage{multirow}
\usepackage{soul}

\usepackage{extarrows}

\color{black}

\def\ba{\mathbf{a}}
\def\bX{\mathbf{X}}
\def\bY{\mathbf{Y}}

\def\blam{\boldsymbol{\lambda}}
\def\I{\mathcal{I}}
\def\RR{\mathcal{R}}

\def\I{\mathcal{I}}
\def\R{\mathbb{R}}

\def\C{\mathbb{C}}
\def\bb{\mathbf{b}}

\def\bQ{\mathbf{Q}}
\def\bq{\mathbf{q}}

\def\Re{\mathcal{R}}
\def\Im{\mathcal{I}}
\def\y{\mathbf{y}}

\def\x{\mathbf{x}}

\def\h{\mathbf{h}}
\def\bF{\mathbf{F}}
\def\bA{\mathbf{A}}
\def\bM{\mathbf{M}}

\def\bH{\mathbf{H}}
\def\bB{\mathbf{B}}

\def\bh{\mathbf{h}}
\def\bthe{\boldsymbol{\Theta}}

\def\bd{\mathbf{d}}

\def\bY{\mathbf{Y}}

\newtheorem{theorem}{Theorem}

\newtheorem{lemma}{Lemma}
\newtheorem{remark}{Remark}

\newtheorem{proposition}{Proposition}

\newtheoremstyle{noparens}%
  {}{}%
  {\itshape}{}%
  {\bfseries}{.}%
  { }%
  {\thmname{#1}\thmnumber{ #2}\mdseries\thmnote{ #3}}

\theoremstyle{noparens}

\date{today}
\title{Finite-Resolution Microwave Linear Analog Computer (MiLAC)-Aided Multiuser Beamforming}
\author{\IEEEauthorblockN{Zheyu Wu, \IEEEmembership{Member, IEEE}, and Bruno Clerckx, \IEEEmembership{Fellow, IEEE}}
  	\thanks{Z. Wu and B. Clerckx are with the Department of Electrical and Electronic Engineering, Imperial College London, London, SW7 2AZ, U.K. (email: \{zheyu.wu, b.clerckx\}@imperial.ac.uk). This work has been partially supported by UKRI grant EP/Y004086/1, EP/X040569/1, EP/Y037197/1, EP/X04047X/1, EP/Y037243/1.}
  }
\date{today}
\begin{document}
\maketitle
\begin{abstract}
Large-scale antenna arrays are essential to future wireless networks but impose significant hardware and digital-processing burdens. The microwave linear analog computer (MiLAC) offers a promising architecture for addressing these challenges. In a MiLAC, beamforming is realized entirely in the analog domain through a reconfigurable multiport microwave network comprising tunable admittance elements. However, existing studies assume that these elements are continuously tunable, whereas practical implementations can support only finitely many configurable states.  This paper investigates finite-resolution MiLAC-aided multiuser beamforming under lossless and reciprocal constraints with prescribed connectivity patterns. Each tunable susceptance is selected from a symmetric uniform finite-resolution codebook. 
We formulate an online problem that jointly optimizes the discrete susceptances and power allocation for a prescribed codebook, and an offline problem that designs the codebook based on channel statistics. 
 By exploiting the odd-uniform codebook structure, we express each multilevel susceptance as a weighted sum of two-level variables and develop exact continuous penalty models for both the online beamforming and offline codebook design problems. We prove that, for sufficiently large penalty parameters, these penalty models are globally equivalent to their original discrete counterparts. Building on these formulations, we develop efficient alternating direction method of multipliers (ADMM)-based algorithms for solving the two problems.
 Numerical results demonstrate that the performance of an unquantized fully-connected MiLAC can be approached with substantially reduced connectivity and finite resolution. In particular, a 3-bit stem-connected MiLAC achieves more than $95\%$ of the sum-rate performance of the unquantized fully-connected MiLAC while requiring only $22\%$ of its tunable components.
\end{abstract}

\section{Introduction}
Toward the 6G era, wireless networks are expected to accommodate substantially higher data traffic and user densities, imposing increasing demands on the spectral efficiency and spatial multiplexing capability. Ultra-massive or gigantic multiple-input multiple-output (MIMO) \cite{6G_emil}, enabled by very large antenna arrays, is envisioned as a key technology to meet these requirements. However, as the array dimension grows, the hardware and signal processing burden of conventional fully digital beamforming becomes prohibitive. In particular, each antenna typically requires a dedicated radio frequency (RF) chain comprising power-hungry components such as high-resolution digital-to-analog converters (DACs), analog-to-digital converters (ADCs), and power amplifiers (PAs), while the complexity of digital beamforming also increases with the array size. This scalability challenge has motivated growing interest in architectures that shift part of the beamforming operation from the digital baseband to the analog domain.

Hybrid digital-analog beamforming is a prominent example of such architectures \cite{hybrid_Ayach,hybrid2,hybrid_survey}. In conventional hybrid beamforming, a small number of RF chains are combined with an analog phase-shifter network to perform beamforming jointly in the digital and analog domains. With only twice as many RF chains as data streams, hybrid beamforming can realize arbitrary fully digital beamforming matrices under ideal phase-shifter assumptions \cite{hybrid2}. 
Other approaches include stacked intelligent metasurfaces (SIMs), which manipulate electromagnetic waves through cascaded programmable layers \cite{an2023stacked}, and dynamic metasurface antennas (DMAs), which perform analog beamforming through tunable radiating elements coupled to waveguides \cite{DMA}.

Recently, the microwave linear analog computer (MiLAC) has been proposed as a new architecture for performing signal processing and computation directly in the analog domain \cite{part1,part2}.  A MiLAC can be modeled as a reconfigurable multiport microwave network composed of tunable impedance or admittance elements. In MiLAC-aided transmitters, its input and output ports are connected to RF chains and antennas, respectively, and tunable impedance or admittance elements determine the beamforming transformation \cite{part2,nerini2026overview}. This architecture requires only as many RF chains as data streams and avoids symbol-level digital beamforming.  Under an idealized model with arbitrarily tunable complex admittances, a MiLAC can realize an arbitrary beamforming matrix and implement specific schemes, such as zero-forcing (ZF), at a significantly lower computational cost than their digital counterparts \cite{part2}. 

Subsequent studies have investigated MiLACs under lossless and reciprocal constraints, which are standard assumptions in microwave network theory \cite{microwavebook}. Losslessness assumes no power dissipation within the network, while reciprocity avoids the need for nonreciprocal components. Together, these constraints yield a purely imaginary and symmetric admittance matrix, which can be parameterized by real-valued susceptances.   The beamforming capabilities of lossless and reciprocal MiLACs have been characterized for both point-to-point and multiuser MIMO systems. Specifically, MiLAC-aided beamforming was shown to achieve the performance of fully digital beamforming in point-to-point MIMO systems \cite{MIMOcapacity}. In multiuser systems, however, a performance gap generally remains, but diminishes as the user channels become more orthogonal  \cite{fang2026,wu2026microwave}. To close this gap, hybrid digital-MiLAC \cite{wu2026microwave} and two-layer MiLAC \cite{zhou2026twolayer} architectures have been proposed, both of which can achieve fully digital beamforming performance in multiuser systems. A stem-connected MiLAC was proposed in \cite{MIMOcapacity2} to reduce the circuit complexity. This architecture can match the performance of a fully-connected MiLAC  \cite{MIMOcapacity2,zhang2026beamforming}, while requiring a number of tunable components that scales linearly with the number of transmit antennas.  Beyond these architectural and beamforming developments, research on MiLAC-aided systems has addressed channel estimation \cite{zhang2026channel,zhang2026channel2}, wideband transmission \cite{peng2026hybrid}, and sensing \cite{liu2026microwave,zhang2026rfchain}. Other studies have characterized the energy efficiency of MiLAC-aided systems \cite{zhang2026quantization,zhou2026lossy} and investigated the effects of hardware characteristics, including component losses \cite{zhou2026lossy}, mutual coupling \cite{nerini2026physics}, and impedance mismatch \cite{xiong2026}.


Despite these advances, most existing studies on MiLAC-aided beamforming assume continuously tunable components. In practical implementations, however, tunable susceptance components can only support a finite number of configurable states. This finite-resolution constraint fundamentally changes the MiLAC-aided beamforming design problem, resulting in a mixed discrete-continuous optimization problem in which the discrete susceptances and continuous power allocation variables jointly determine the achievable sum rate. The problem is further complicated by the nonlinear dependence of the beamforming matrix on the susceptances through a matrix inverse. In addition to optimizing the discrete susceptances for a given codebook, the codebook itself must also be carefully designed. For a prescribed number of quantization bits, it determines the available susceptance levels and the overall dynamic range, both of which affect the achievable  performance.  
Effective finite-resolution MiLAC design therefore requires both online optimization of the network configuration for each channel realization and offline design of the codebook according to the channel statistics.

A very recent work \cite{zhang2026beamforming} has taken an initial step toward finite-resolution MiLAC-aided  beamforming  design by first optimizing the susceptances under continuous tuning and then projecting them onto a predefined finite-resolution codebook. An alternating refinement procedure was further proposed to improve the projected solution by updating the susceptances one at a time. However, such a sequential approach can be sensitive to the initialization. The reported numerical results show that the sum rate gradually saturates as the number of quantization bits increases, with a noticeable performance gap from the unquantized benchmark.  It remains unclear whether the gap is due to the inherent limitations of finite resolution or the suboptimality of the optimization algorithm. Moreover, the quantization range and grid are predefined rather than optimized.

Motivated by these limitations, this paper focuses on finite-resolution MiLAC-aided multiuser systems and investigates both online beamforming design and offline codebook design.  We consider a lossless and reciprocal MiLAC and accommodate general MiLAC architectures by setting the susceptances associated with disconnected port pairs to zero. Each remaining tunable susceptance is constrained to an odd-uniform finite-resolution codebook symmetric about zero. 
 The main contributions are summarized as follows.

1) \emph{Finite-resolution MiLAC-aided beamforming model:}  We formulate the finite-resolution MiLAC-aided sum-rate maximization problem for multiuser multiple-input single-output (MISO) systems under lossless and reciprocal constraints, while accommodating general MiLAC architectures through prescribed connectivity patterns of the underlying microwave network. For a fixed codebook, the online problem jointly optimizes the tunable susceptances and the power-allocation vector according to the instantaneous channel realization. The offline problem optimizes a shared quantization parameter, which determines the codebook spacing and dynamic range, according to the channel statistics.

2) \emph{Online beamforming design:} For a given finite-resolution codebook, we develop an efficient optimization framework for the discrete MiLAC configuration. By exploiting the odd-uniform structure of the codebook, we express each multilevel susceptance as a weighted sum of two-level variables.  We then develop a continuous quadratic penalty model and prove that it is globally equivalent to the original discrete problem when the penalty parameter is sufficiently large.  Based on this exact penalty model, we develop a penalty-continuation algorithm together with an ADMM-based algorithm for efficiently solving the resulting nonconvex problem.

3) \emph{Offline codebook design:} We further optimize the finite-resolution codebook according to the channel statistics.  Using sample average approximation, we formulate an offline design problem in which a shared quantization parameter is jointly optimized over multiple training channel realizations. We extend the exact penalty framework to this setting and prove that the resulting continuous penalty problem remains globally equivalent to the original discrete sample average problem for a sufficiently large penalty parameter. An efficient ADMM-based algorithm is then developed, in which the shared codebook parameter is updated jointly across the training samples using a proximal gradient algorithm.

4) \emph{Performance evaluation and design insights:}
Extensive numerical results validate the effectiveness of the proposed online and offline designs. The proposed online beamforming algorithm achieves much higher sum rates than the methods in \cite{zhang2026beamforming}, and the offline codebook design achieves performance close to that obtained by exhaustive search. Furthermore, the results reveal the effects of quantization resolution and circuit connectivity of MiLAC on the achievable performance. In particular, a 3-bit stem-connected MiLAC can approach the performance of an unquantized fully-connected MiLAC while requiring substantially fewer tunable interconnections.

\emph{Organization:}
The rest of this paper is organized as follows. Section \ref{sec:2} introduces the MiLAC-aided multiuser MISO system model, describes the considered hardware constraints, and formulates the online beamforming and offline codebook-design problems. Sections \ref{sec:online} and \ref{sec:offline} develop the proposed algorithms for solving the online and offline problems, respectively.   Section \ref{sec:simulation} presents numerical results to evaluate the proposed algorithms and investigate the effects of quantization resolution and circuit connectivity. Finally, Section \ref{sec:conclusion} concludes the paper. 

\emph{Notations:}
Throughout the paper, scalars, column vectors, matrices, and sets are denoted by lowercase letters, bold lowercase letters, bold uppercase letters, and calligraphic letters, respectively.  For a matrix
$\mathbf X$, $[\mathbf X]_{i,j}$ denotes its $(i,j)$th entry, while $[\mathbf X]_{\mathcal I_1,\mathcal I_2}$ denotes the submatrix formed by the rows indexed by $\mathcal I_1$ and the columns indexed by
$\mathcal I_2$. The operators $(\cdot)^T$, $(\cdot)^H$, and
$(\cdot)^{-1}$ denote the transpose, Hermitian transpose, and matrix inverse, respectively. The operators $\Re(\cdot)$ and $\Im(\cdot)$ denote
the real and imaginary parts, respectively,  and $\mathbb{E}[\cdot]$ denotes the expectation operator. For a vector $\mathbf x$,
$\operatorname{diag}(\mathbf x)$ denotes the diagonal matrix whose
diagonal entries are given by $\mathbf x$. The operators
$\operatorname{vec}(\cdot)$ and $\operatorname{blkdiag}(\cdot)$ denote
vectorization and block-diagonal concatenation, respectively. The norms $\|\cdot\|_2$, $\|\cdot\|_F$, and $\|\cdot\|_\infty$ denote
the Euclidean/spectral norm, Frobenius norm, and infinity norm,
respectively. The symbols $\mathbf I$ and $\mathbf 0$ denote the
identity matrix and the all-zero matrix, with dimensions specified by
subscripts when necessary.
\section{System Model and Problem Formulation}\label{sec:2}
\subsection{System Model}
\begin{figure}
\includegraphics[width=0.5\textwidth]{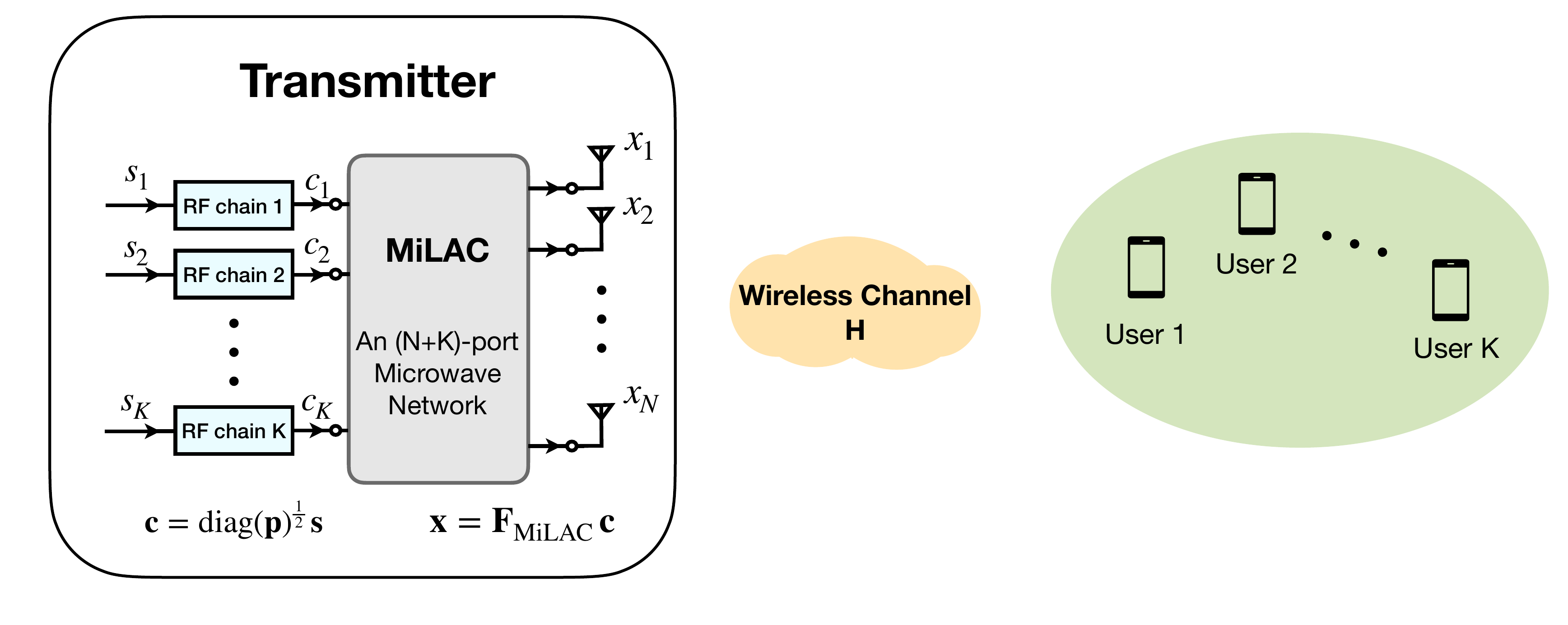}
\caption{A MiLAC-aided multiuser MISO system, where $\mathbf F_{\mathrm{MiLAC}}\in\mathbb C^{N\times K}$ characterizes the input-output relationship of the MiLAC.}
\centering
\label{fig:system}
\end{figure}
Consider a MiLAC-aided downlink multiuser MISO system, where an $N$-antenna transmitter simultaneously serves $K$ single-antenna users. As shown in Fig.~\ref{fig:system}, the transmitter is equipped with $K$ RF chains and an $(N+K)$-port MiLAC. The $K$ input ports of the MiLAC are connected to the $K$ RF chains, while the remaining $N$ output ports are connected to the $N$ transmit antennas.

Let $\mathbf{s}=[s_1,s_2,\ldots,s_K]^T$ denote the data symbol vector intended for the $K$ users, with $\mathbb{E}[\mathbf{s}\mathbf{s}^H]=\mathbf I_K$. After power allocation, the RF-chain output signal is given by $$    \mathbf c = \operatorname{diag}(\sqrt{\mathbf p})\mathbf s, $$ where $\mathbf p=[p_1,p_2,\ldots,p_K]^T$ denotes the power allocation vector. The signals are then fed into the MiLAC, which realizes an analog beamformer between the RF-chain outputs and the antenna ports,
yielding the transmit signal
 $$    \mathbf x
    =
    \mathbf F_{\mathrm{MiLAC}}\mathbf c
    =
    \mathbf F_{\mathrm{MiLAC}}
    \operatorname{diag}(\sqrt{\mathbf p})\mathbf s.
$$Here,  $\mathbf F_{\mathrm{MiLAC}}\in\mathbb C^{N\times K}$ denotes the
MiLAC-induced analog beamforming matrix, which is determined by the circuit topology and tunable admittances inside the network. Its explicit expression will be described in detail in Section \ref{subsec:Fmilac}.

Accordingly, the received signal at user $k$ is given by
\begin{equation}\label{eq:system}
    r_k= \mathbf h_k^H \mathbf x + n_k = \mathbf h_k^H   \mathbf F_{\mathrm{MiLAC}}
    \operatorname{diag}(\sqrt{\mathbf p})\mathbf s     + n_k,
\end{equation}
where $\mathbf h_k\in\mathbb C^{N\times 1}$ denotes the channel from the
transmitter to user $k$, and $n_k\sim\mathcal{CN}(0,\sigma_0^2)$ denotes the additive noise.

\subsection{MiLAC Beamforming Matrix}\label{subsec:Fmilac}
We now discuss the MiLAC beamforming matrix $\bF_{\mathrm{MiLAC}}$. 
A MiLAC is a multiport microwave network composed of tunable admittance (impedance) elements. Fig.~\ref{fig:MiLAC} illustrates a four-port fully-connected MiLAC with \(K=2\)
input ports and \(N=2\) output ports. 
 \begin{figure}
\includegraphics[width=0.3\textwidth]{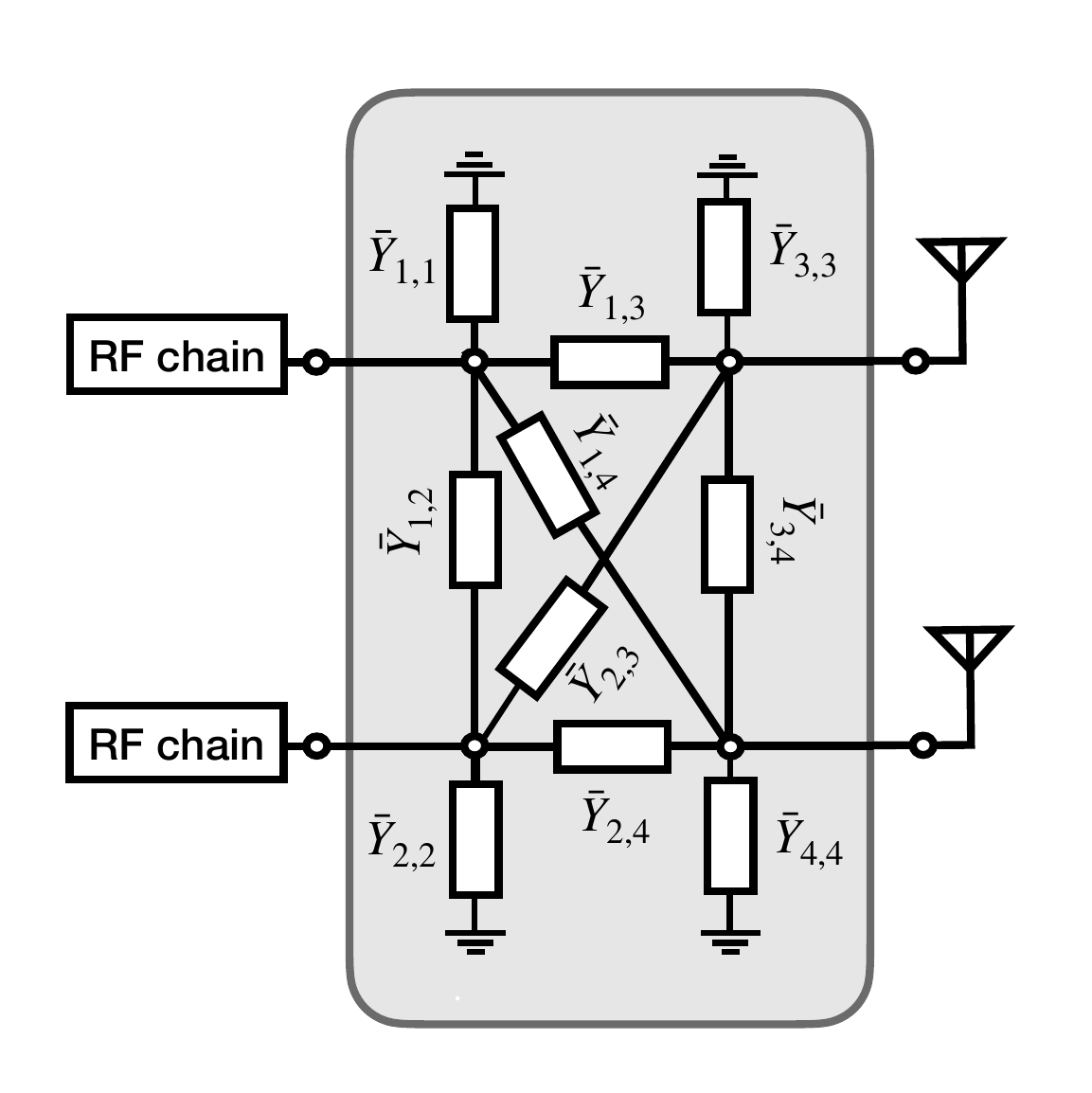}
\centering
\caption{A four-port fully-connected MiLAC with $K=2$ input ports and $N=2$ output ports.}
\label{fig:MiLAC}
\end{figure} 
 Following \cite{part1,part2}, the MiLAC-induced beamforming matrix $\bF_{\mathrm{MiLAC}}$ can be expressed in terms of the admittance matrix    $\bY\in\C^{(N+K)\times (N+K)}$ of the underlying microwave network as    
\begin{equation}\label{Fmilac:Y}
\mathbf{F}_{\text{MiLAC}}=\left[\left(\mathbf{I}_{N+K}+{Z_0\mathbf{Y}}\right)^{-1}\right]_{K+1:K+N,1:K},
\end{equation}
where $Z_0$ denotes the reference impedance, typically set to $50 \, \Omega$.
The admittance matrix $\bY$  is determined by the circuit topology and tunable admittance elements.   Let \(\bar{Y}_{i,j}\) denote the admittance connecting the \(i\)-th and \(j\)-th ports for \(i\neq j\), and let \(\bar{Y}_{i,i}\) denote the  admittance connecting the \(i\)-th port to ground, as shown in Fig. \ref{fig:MiLAC}. Then the entries of the admittance matrix \(\mathbf{Y}\) are given by
\begin{equation}\label{admittance}
[\bY]_{i,j}=\left\{
\begin{aligned}
&-\bar{Y}_{i,j},~~~~~~~~~\,~~&\text{if}~i\neq j;\\
&\bar{Y}_{i,i}+\textstyle\sum_{n\neq i}\bar{Y}_{i,n},~~&\text{if }i=j. 
\end{aligned}\right.
\end{equation}
In particular, $\bar{Y}_{i,j}=0$ if the $i$-th and $j$-th ports are not connected.

As shown in \cite{MIMOcapacity}, the MiLAC-induced beamforming matrix can also be represented using the scattering matrix of the microwave network. Let \(\bthe\in\mathbb{C}^{(N+K)\times(N+K)}\) denote the scattering matrix of the $(N+K)$-port MiLAC. Then,
\begin{equation}\label{F:theta}
\bF_{\mathrm{MiLAC}}=\frac{1}{2}\left[\bthe\right]_{K+1:N+K,1:K}.
\end{equation}
 The expression in \eqref{F:theta} follows from \eqref{Fmilac:Y} and the standard relationship between the scattering matrix and the admittance matrix of a multiport microwave network \cite{microwavebook}, namely 
 \begin{equation}\label{Eq:thetaY}
\bthe
=
\left(\mathbf I_{N+K}+Z_0\bY\right)^{-1}
\left(\mathbf I_{N+K}-Z_0\bY\right).
\end{equation}

\subsection{Hardware Constraints on MiLAC}
We next introduce several hardware constraints on MiLAC considered in this paper.

\subsubsection{Lossless and reciprocal constraint}The admittance components of the MiLAC are typically assumed to be lossless and reciprocal \cite{MIMOcapacity,MIMOcapacity2,fang2026,wu2026microwave,zhou2026twolayer,zhang2026beamforming}. The lossless constraint ensures that the network incurs no power dissipation. In this case, the admittance components are purely imaginary and can be expressed as
$$
\bar{Y}_{i,j}=\mathrm{i}\bar{B}_{i,j},
$$
where \(\bar{B}_{i,j}\in\R\) denotes the corresponding susceptance. Reciprocity means that the transmission behavior between any two ports is identical in both directions, thereby avoiding the need for complicated nonreciprocal elements. Accordingly, the inter-port susceptances satisfy
$$
\bar{B}_{i,j}=\bar{B}_{j,i}, \quad \forall~ i\neq j.
$$
The admittance matrix can then be expressed as 
\begin{equation}\label{Eq:YB}
\mathbf{Y}=\mathrm{i}\mathbf{B}, 
\end{equation}
where \(\mathbf{B}\in\mathbb{R}^{(N+K)\times(N+K)}\) is a real symmetric susceptance matrix.

\subsubsection{Architecture-induced sparsity constraint}\label{subsubsec:arch}
In general, a MiLAC architecture specifies which pairs of ports are physically interconnected through tunable admittance elements. Hence, different architectures can be represented by different sparsity patterns on the susceptances $\{\bar{B}_{i,j}\}$. Specifically, let \(\mathcal D\) denote the set of disconnected port pairs, i.e.,
\begin{equation}\label{def:calA}
(i,j)\in\mathcal D~ \Longleftrightarrow ~
\text{ports } i \text{ and } j \text{ are disconnected}, ~ i\neq j,
\end{equation}
then
\[
\bar B_{i,j}=0,\quad \forall~ (i,j)\in\mathcal D.
\]
Due to the reciprocity of the network, $(i,j)\in\mathcal{D}$ if and only if $(j,i)\in\mathcal{D}$. 
In contrast, if \((i,j)\notin\mathcal D\) and \(i\neq j\), ports \(i\) and \(j\) are connected by a tunable admittance.

The most common MiLAC architecture studied in the literature is the fully-connected MiLAC, where every pair of ports is interconnected, and thus  \(\mathcal D=\emptyset\). An example is given by  Fig. \ref{fig:MiLAC}. 
To reduce the circuit complexity, a novel stem-connected architecture has been proposed in \cite{MIMOcapacity2}. In a stem-connected MiLAC, a subset of ports is selected as central ports. These central ports are connected to all other ports, while the remaining non-central ports are only connected to the central ports and have no interconnections among themselves. Hence,
\[
\mathcal D=\{(i,j)\mid  i \text{ and } j \text{ are both non-central ports},~i\neq j\}.
\]
It has been shown in \cite{MIMOcapacity2} that, under the lossless and reciprocal constraints with continuously tunable susceptances, a stem-connected MiLAC with \(2K-1\) stem ports achieves the performance of the fully-connected MiLAC. Throughout this paper, we assume that at least one input port is  connected to an output port. Otherwise, no signal can be transferred from the RF-chain side to the antenna side. 
\subsubsection{Finite-resolution constraint}
The above architectures are usually studied under the assumption that the available susceptances \(\{\bar B_{i,j}\}\) are continuously tunable. In practice, however, these susceptances  can only take values from a finite discrete set. Motivated by this practical limitation, we  consider a finite-resolution setting in which each tunable susceptance is selected from a finite codebook \(\mathcal C\). Specifically, we adopt a \(b\)-bit codebook symmetric around zero with uniform quantization, given by
$$\mathcal{C}=\{\pm c,\pm 3c,\dots,\pm (2^{b}-1)c\},$$ 
where $c$ denotes half of the quantization stepsize.
\subsection{Problem Formulation}
In this paper, the MiLAC architecture is assumed to be specified a priori and is not treated as the optimization variable. Our focus is to optimize the tunable susceptances and the finite-resolution codebook under a given architecture.

The goal is twofold.  First, for a given finite-resolution codebook $\mathcal{C}$, the MiLAC susceptances $\{\bar{B}_{i,j}\}$ and the power allocation vector $\mathbf{p}$ are jointly optimized online according to the instantaneous channel realization. Second, the quantization parameter $c$, which determines the codebook, is designed offline based on the channel statistics. In the following, we first formulate the online beamforming problem for a fixed $c$, and then present the corresponding offline design problem for $c$.

\emph{1) Online Problem: Finite-Resolution MiLAC Beamforming.} 
For notational simplicity, define $\bF=2\bF_{\text{MiLAC}}$. We use the sum-rate as the performance metric, which, according to \eqref{eq:system}, reads $$R(\mathbf{p},\bF; \bH)=\sum_{k=1}^K\log\left(1+\frac{|\h_k^H\mathbf{f}_k|^2p_k}{\sum_{j\neq k}|\h_k^H\mathbf{f}_j|^2p_j+\sigma^2}\right),$$
where $\mathbf{f}_k$ denotes the $k$-th column of $\mathbf{F}$ and $\sigma^2=4\sigma_0^2$ absorbs the factor of $2$ between $\bF$ and $\bF_{\text{MiLAC}}$ into the noise variance. Here, the notation $R(\mathbf{p},\bF;\bH)$ is used to emphasize that the achievable sum-rate is a function of the beamforming matrix $\bF$ and the power allocation vector $\mathbf{p}$, under a given channel realization $\bH=[\h_1,\h_2,\dots,\h_K]^H\in\C^{K\times N}$. 

For a given MiLAC architecture, let $\mathcal D$ denote the set of disconnected port pairs as defined in \eqref{def:calA}. We further define
$$\mathcal E\triangleq
\{(i,j)\mid 1\leq i, j\leq N+K\}\setminus \mathcal D,$$
which collects all self-pairs and connected port pairs, whose corresponding susceptances $\bar B_{i,j}$ are tunable. Given the channel realization \(\mathbf H\) and a fixed quantization parameter \(c\), the online finite-resolution MiLAC beamforming problem is formulated as

  \begin{subequations}\label{problem:sumrate}
\begin{align}
R^\star(c;\bH)~~~~~&\notag\\
\triangleq\max_{\mathbf{p},\bF,\bthe,\bB,\bar{\bB}}~&R(\mathbf{p},\bF;\bH)\\
\text{s.t.}~~~~~ &\bF=[\bthe]_{K+1:K+N,1:K}, \label{con:F}\\&
\bthe=\left(\mathbf{I}_{N+K}+\mathrm{i}Z_0\mathbf{B}\right)^{-1}\left(\mathbf{I}_{N+K}-\mathrm{i}Z_0\mathbf{B}\right),\label{con:theta_B}\\
~~&[\bB]_{i,j}=\left\{
\begin{aligned}
&-\bar{B}_{i,j},~~~~~~~~~\,~~&\text{if}~i\neq j;\\
&\bar{B}_{i,i}+\textstyle\sum_{n\neq i}\bar{B}_{i,n},~~&\text{if }i=j.
\end{aligned}\right.\label{con:B_barB}\\
&\bar{B}_{i,j}=0,~\forall~(i,j)\in {\cal D},\label{con:arch}
\\&\bar{B}_{i,j}\in\mathcal{C},~\forall~(i,j)\in {\cal E},\label{con:discrete}\\&\bar{B}_{i,j}=\bar{B}_{j,i},~\forall~(i,j)\in{\cal E},\label{con:reciprocal}\\
&\mathbf{1}^T\mathbf{p}\leq P_T,~\mathbf{p}\geq \mathbf{0}.\label{con:power}
\end{align}
\end{subequations}
In the above,  \eqref{con:F} follows from \eqref{F:theta} and the definition $\bF=2\bF_{\mathrm{MiLAC}}$, \eqref{con:theta_B}  is obtained from the relationship between the scattering matrix $\bthe$ and the admittance matrix $\bY$ in  \eqref{Eq:thetaY} and the lossless constraint that yields \eqref{Eq:YB}, \eqref{con:B_barB} follows from the construction of the susceptance matrix $\mathbf B$ as specified  in \eqref{admittance}, \eqref{con:arch} --  \eqref{con:reciprocal}  
impose the architecture-induced sparsity constraint,
the finite-resolution codebook constraint, and the reciprocity constraint, on the susceptances $\{\bar{B}_{i,j}\}$, respectively,
 and \eqref{con:power} represents the total transmit power constraint, where $P_T$ denotes the maximum transmit power.

The optimal sum-rate is  a function of the quantization parameter $c$ and the channel matrix $\bH$, thus written as $R^\star(c;\bH)$.  This problem describes the online optimization stage, where the codebook is fixed while the MiLAC susceptances and the power allocation are adapted to the instantaneous channel.

\emph{2) Offline Problem: Codebook  Design. } 
The quantization parameter $c$ determines the available susceptance levels in the finite-resolution codebook. Since $c$ is a hardware-related parameter, it is designed offline based on the channel distribution as
\begin{equation}\label{problem:codebook}
\max_{c\geq 0} ~\mathbb{E}_\bH[R^\star(c;\bH)].
\end{equation} 

In the following Sections \ref{sec:online} and \ref{sec:offline}, we solve the online problem in \eqref{problem:sumrate} and the offline problem in \eqref{problem:codebook}, respectively. 

\section{Algorithm for Solving the Online Sum-Rate Maximization Problem}\label{sec:online}
 Problem \eqref{problem:sumrate} is a challenging non-convex optimization problem with  mixed discrete and continuous variables. Since the number of tunable susceptances is typically large, exhaustive search over all possible discrete configurations is infeasible. In this section, we develop an efficient algorithm for solving the online problem in \eqref{problem:sumrate}.

First, in Section \ref{subsec:reformulation}, we reformulate problem \eqref{problem:sumrate} into a more tractable form by introducing appropriate auxiliary variables. Second, in Section \ref{subsec:penalty}, we develop a continuous penalty model of the discrete problem and show that, with a sufficiently large penalty parameter, it is globally equivalent to the original discrete problem. Finally, we propose an ADMM-based algorithm to efficiently solve the resulting continuous problem in Section \ref{subsec:ADMM}.
\subsection{Reformulation of \eqref{problem:sumrate}}\label{subsec:reformulation}
We reformulate problem \eqref{problem:sumrate} in two steps. First, we eliminate the matrix inverse constraint by introducing an auxiliary beamforming variable. Second, we exploit the odd-uniform structure of the finite-resolution codebook to represent each discrete susceptance using several two-level variables.

To begin, define two selection matrices as 
 $$\bar{\mathbf{I}}_1=\left[\begin{matrix} \mathbf{I}_K\\\mathbf{0}_{N\times K}\end{matrix}\right]~\text{and }~\bar{\mathbf{I}}_2=[\mathbf{0}_{N\times K}~~\mathbf{I}_N].$$ Then, constraints  \eqref{con:F}  and \eqref{con:theta_B} can be rewritten as 
 \begin{equation}\label{con:BF}
 \bar{\mathbf{I}}_2\left(\mathbf{I}_{N+K}+\mathrm{i}Z_0\mathbf{B}\right)^{-1}\left(\mathbf{I}_{N+K}-\mathrm{i}Z_0\mathbf{B}\right)\bar{\mathbf{I}}_1=\bF, 
 \end{equation}
 which eliminates the scattering matrix $\bthe$. However, the resulting constraint still contains a matrix inverse. To remove this inverse, we introduce an auxiliary variable  $\bF_c\in\C^{K\times K}$ and define $$\bar{\bF}=\left[\begin{matrix} \bF_c\\\bF\end{matrix}\right]\in\C^{(N+K)\times K}.$$
Then $\bF$ can be expressed as  $\bF=\bar{\mathbf{I}}_2\bar{\bF}$, and  the sum-rate can be written as a function of $\bar{\bF}$ and $\mathbf{p}$: 
$$R(\mathbf{p},\bar{\bF}; \bH)=\sum_{k=1}^K\log_2\left(1+\frac{|\h_k^H\bar{\mathbf{I}}_2\bar{\mathbf{f}}_k|^2p_k}{\sum_{j\neq k}|\h_k^H\bar{\mathbf{I}}_2\bar{\mathbf{f}}_j|^2p_j+\sigma^2}\right).$$
Moreover, by the definition of $\bar{\bF}$, \eqref{con:BF} can be represented as  
$$\left(\mathbf{I}_{N+K}+\mathrm{i}Z_0\mathbf{B}\right)^{-1}\left(\mathbf{I}_{N+K}-\mathrm{i}Z_0\mathbf{B}\right)\bar{\mathbf{I}}_1=\bar{\bF},$$
or equivalently,
 $$\mathrm{i}Z_0\bB(\bar{\bF}+\bar{\mathbf{I}}_1)=\bar{\mathbf{I}}_1-\bar{\bF}.$$
This reformulation removes the matrix inverse and leads to a bilinear constraint in $\bar{\bF}$ and $\bB$, which is more convenient for algorithm design.

We next reformulate the finite-resolution constraint in \eqref{con:discrete}. By exploiting the odd-uniform structure of the codebook $\mathcal{C}$, each feasible susceptance value $\bar{B}_{i,j}\in\mathcal{C}$ can be represented as a weighted sum of $b$ two-level variables:
\[
\bar{B}_{i,j}= \sum_{\ell=1}^b 2^{\ell-1} \bar{B}_{i,j}^{(\ell)},~ \bar{B}_{i,j}^{(\ell)}\in\{-c,c\},~\ell=1,2,\dots, b.
\]
Therefore, the multi-level finite-resolution constraint is equivalently replaced by a set of two-level discrete constraints. This representation leads to a more structured formulation and facilitates the development of the penalty-based algorithm in the following subsection.

Using the above transformations, problem \eqref{problem:sumrate} can be equivalently rewritten as
  \begin{subequations}\label{problem2:sumrate}
\begin{align}
\max_{\mathbf{p},\bar{\bF},\bB,\atop\{\bar{B}_{i,j}\},\{\bar{B}_{i,j}^{(\ell)}\}}
&R(\mathbf{p},\bar{\bF};\bH)\\
\text{s.t.}~~~~~~ &\mathrm{i}Z_0\bB(\bar{\bF}+\bar{\mathbf{I}}_1)=\bar{\mathbf{I}}_1-\bar{\bF},\label{con:BbarF2}\\
~~&[\bB]_{i,j}=\left\{
\begin{aligned}
&-\bar{B}_{i,j},~~~~~~~~~\,~&\text{if}~i\neq j;\\
&\bar{B}_{i,i}+\textstyle\sum_{n\neq i}\bar{B}_{i,n},~&\text{if }i=j.
\end{aligned}\right. \label{con:BbarB2}\\
&\bar{B}_{i,j}=0,~\forall~(i,j)\in\mathcal{D},\\
&\bar{B}_{i,j}= \sum_{\ell=1}^b 2^{\ell-1} \bar{B}_{i,j}^{(\ell)},~\forall~(i,j)\in\mathcal{E},\label{con:BBl2}\\
~&\bar{B}_{i,j}^{(\ell)}\in\{-c,c\},~\forall~(i,j)\in \mathcal{E},~\ell=1,\dots,b, \label{con:binary}\\
&\bar{B}_{j,i}^{(\ell)}=\bar{B}_{i,j}^{(\ell)},~\forall~(i,j)\in \mathcal{E},~\ell=1,\dots,b,\label{con:reciprocal2}\\
&\mathbf{1}^T\mathbf{p}\leq P_T,~\mathbf{p}\geq \mathbf{0}\label{con:power2}.
\end{align}
\end{subequations}
Compared with the original formulation in \eqref{problem:sumrate}, the reformulated problem eliminates the matrix inverse constraint by introducing the auxiliary beamforming variable $\bar{\bF}$. In addition, the finite-resolution constraint is equivalently represented by a set of two-level discrete constraints on $\{\bar{B}_{i,j}^{(\ell)}\}$. In the following subsection, we further develop a continuous penalty model to handle these two-level constraints.

\subsection{Penalty Model of Problem \eqref{problem2:sumrate}}\label{subsec:penalty}
In this subsection, we develop a continuous penalty model to deal with the two-level discrete constraints in \eqref{con:binary}. 

For each \((i,j)\in\mathcal E\) and \(\ell=1,\ldots,b\),  we first relax the two-level  constraint to its convex hull, i.e.,
$
\bar B_{i,j}^{(\ell)}\in[-c,c],
$
and then promote the two-level structure through a penalty term. 
 Specifically, define the penalty function as
\begin{equation}\label{eq:penalty}
 P(\bar{B}_{i,j}^{(\ell)})=\left(\bar B_{i,j}^{(\ell)}\right)^2-c^2.%
\end{equation}
 Under the box constraint \(\bar B_{i,j}^{(\ell)}\in[-c,c]\), we have
\[
  P(\bar{B}_{i,j}^{(\ell)})\leq 0,
\]
where equality holds if and only if
\[
\bar B_{i,j}^{(\ell)}\in\{-c,c\},
\quad \forall (i,j)\in\mathcal E,\ \ell=1,\ldots,b.
\]
Therefore, maximizing $P(\bar{B}_{i,j}^{(\ell)})$ drives the relaxed variable toward the desired two-level discrete set  $\{-c,c\}$. 

Collecting the penalty terms for all two-level variables, we define
$$P(\{\bar B_{i,j}^{(\ell)}\})=\sum_{\ell=1}^{b}\sum_{\substack{(i,j)\in\mathcal E, j\geq i}}
P(\bar B_{i,j}^{(\ell)}),$$ 
where the penalty is applied only to the upper-triangular entries with $j\geq i$ due to the symmetry constraint in \eqref{con:reciprocal2}. The resulting continuous penalty model is formulated as  
\begin{subequations}\label{problem2:sumrate-penalty}
\begin{align}
\max_{\mathbf{p},\bar{\bF},\bB,\atop\{\bar{B}_{i,j}\},\{\bar{B}_{i,j}^{(\ell)}\}}
&R(\mathbf{p},\bar{\bF};\bH)+\frac{\rho}{2} P(\{\bar B_{i,j}^{(\ell)}\})\\
\text{s.t.}~~~~~ &\eqref{con:BbarF2}-\eqref{con:BBl2},\, \eqref{con:reciprocal2}-\eqref{con:power2},\\
~&\bar{B}_{i,j}^{(\ell)}\in[-c,c],~\forall~(i,j)\in \mathcal{E},~\ell=1,\dots,b.  \label{con:box}
\end{align}
\end{subequations}
In \eqref{problem2:sumrate-penalty}, the  two-level discrete constraints are relaxed to box constraints, and an additional penalty term is added to the objective function to promote discreteness, where $\rho>0$ is the penalty parameter.

We have transformed the original discrete problem in \eqref{problem2:sumrate} into the penalized continuous problem in \eqref{problem2:sumrate-penalty}. A natural question is whether the proposed penalty model is exact, namely, whether the penalized continuous problem becomes equivalent to the original discrete problem when the penalty parameter is sufficiently large. The following theorem gives an affirmative answer.

\begin{theorem}\label{equivalence}
There exists a constant $\bar{\rho}>0$ such that, for any $\rho>\bar{\rho}$, problems  \eqref{problem2:sumrate} and \eqref{problem2:sumrate-penalty} are globally equivalent. Specifically, any optimal solution to problem \eqref{problem2:sumrate} is also optimal to problem \eqref{problem2:sumrate-penalty}. Conversely, any optimal solution to problem \eqref{problem2:sumrate-penalty} satisfies the two-level discrete constraints in \eqref{con:binary}, and is also optimal to problem \eqref{problem2:sumrate}.
\end{theorem}
\begin{proof}
See Appendix \ref{proof:equivalence}.
\end{proof}
\begin{remark}
The quadratic penalty term in \eqref{eq:penalty} has also been adopted in existing works to promote discreteness \cite{onebit_shao,onebit_wu}. However, the exactness of the corresponding penalty model is generally problem dependent and requires a case-by-case analysis. For example, the problems considered in \cite{onebit_shao} involve only discrete constraints, while the problem in \cite{onebit_wu} contains only linear constraints in addition to the discrete constraints. In contrast, in the considered MiLAC beamforming problem, the two-level discrete variables are coupled through the constraints in \eqref{con:BbarF2}-\eqref{con:BBl2}, including the nonlinear bilinear constraint \eqref{con:BbarF2}.  Therefore, the considered problem is more complicated than the models in  \cite{onebit_shao,onebit_wu}, and  the exactness of the proposed penalty model requires a dedicated proof. Theorem \ref{equivalence} establishes this exactness by carefully exploiting the specific structure of the problem constraints.
\end{remark}

For the considered problem, it is difficult to obtain the explicit value of the threshold $\bar{\rho}$ in Theorem \ref{equivalence}. Directly setting the penalty parameter $\rho$ to a very large value may lead to poor local solutions. Therefore, we adopt a continuation strategy \cite{onebit_shao,onebit_wu} in practice. The algorithm starts from a moderate penalty parameter $\rho_0$ and gradually increases it until a feasible discrete solution is obtained. The continuation procedure is summarized in Algorithm \ref{alg:rho_update}.  To improve the numerical performance, we also construct a feasible discrete candidate after each penalty stage by rounding the relaxed variables $\{\bar{B}_{i,j}\}$ to the nearest points in $\mathcal{C}$.  The sum-rate objective is then evaluated. Among all discrete candidates generated during the continuation process, we select the one that achieves the largest sum-rate as the final output.

In the following subsection, we propose an efficient algorithm for solving the continuous penalty model in \eqref{problem2:sumrate-penalty}. 
\begin{algorithm}[t] 
\caption{Continuation Framework for Solving Problem \eqref{problem2:sumrate}} 
\label{alg:rho_update}
 \begin{algorithmic}[1] 
 \State \textbf{Input:} Initial penalty parameter $\rho_0$ and penalty growth factor $\tau>1$. 
\State Initialize $\rho=\rho_0$ and set $R_{\rm best}=-\infty$. 
\While{the discrete constraint in \eqref{con:binary} is not satisfied}
 \State Solve the continuous penalty problem \eqref{problem2:sumrate-penalty} with fixed $\rho$ using Algorithm \ref{alg:online_admm} developed in Section \ref{subsec:ADMM}.
  \State Round $\{\bar{B}_{i,j}\}_{(i,j)\in\mathcal{E}}$ to the discrete set $\mathcal{C}$ and construct $\mathbf B_{\rm rd}$ by \eqref{con:BbarB2}. 
\State Set $\bar{\mathbf F}_{\rm rd}=(\mathbf I+\mathrm{i}Z_0\mathbf B_{\rm rd})^{-1}(\mathbf I-\mathrm{i}Z_0\mathbf B_{\rm rd})\bar{\mathbf I}_1$.
 \State With $\bar{\mathbf F}_{\rm rd}$ fixed, refine $\mathbf p_{\rm rd}$  by iteratively applying the updates in \eqref{update:gamma}-\eqref{update:pk} of Section \ref{subsec:ADMM} until convergence. 
 \State Set $R_{\rm rd}=R(\mathbf p_{\rm rd},\bar{\mathbf F}_{\rm rd};\mathbf H)$.
 \If{$R_{\rm rd}>R_{\rm best}$} \State Store the rounded discrete solution. \State Set $R_{\rm best}=R_{\rm rd}$. 
 \EndIf
  \State Increase the penalty parameter: $\rho\leftarrow\tau\rho$.  
 \State Use the current solution as the initialization for the next penalty stage.
  \EndWhile
 \State \textbf{Output}: The stored discrete solution that achieves $R_{\rm best}$.
  \end{algorithmic}
   \end{algorithm}

\subsection{Algorithm for Solving the Penalty Model in \eqref{problem2:sumrate-penalty}}\label{subsec:ADMM}
This subsection focuses on solving the continuous penalty model in \eqref{problem2:sumrate-penalty} for a given penalty parameter $\rho$. The main challenges of problem \eqref{problem2:sumrate-penalty} lie in the non-convex sum-rate objective function and the bilinear constraint in \eqref{con:BbarF2}, which couples variables $\bar{\mathbf F}$ and ${\mathbf B}$.  The other constraints are all linear and can be handled more directly.  To address these challenges, we first transform the objective function using the fractional programming (FP) technique \cite{FP,FP2}, and then develop an ADMM-based algorithm to handle the bilinear constraint.


We now detail the proposed algorithm. Applying the fractional programming technique in \cite{FP,FP2}, the sum rate  objective function can be equivalently transformed to 
$$\max_{\boldsymbol{\gamma},\y} \tilde{R}(\boldsymbol{\gamma},\y,\mathbf{p},\bar{\bF};\bH),$$
where 
$$
\begin{aligned}
& \tilde{R}(\boldsymbol{\gamma},\y,\mathbf{p},\bar{\bF};\bH)\\
&=\sum_{k=1}^K\log(1+\gamma_k)-\gamma_k+(1+\gamma_k)\bigg(2\mathcal{R}(y_k^*\h_k^H\bar{\mathbf{I}}_2\bar{\mathbf{f}}_k)\sqrt{p_k}\\
&~~~-|y_k|^2\big(\sum_{j=1}^K|\h_k^H\bar{\mathbf{I}}_2\bar{\mathbf{f}}_j|^2p_j+\sigma^2\big)\bigg).
\end{aligned}
$$
The advantage of this reformulation is that, by introducing the auxiliary variables  $(\boldsymbol \gamma,\y)$, the objective function becomes separately concave in each variable block, which   facilitates efficient block-wise optimization. To handle the bilinear constraint in \eqref{con:BbarF2}, we further adopt an ADMM framework. Specifically, the augmented Lagrangian function is given by 
$$
\begin{aligned}L(\mathcal{V})=& \tilde{R}(\boldsymbol{\gamma},\y,\mathbf{p},\bar{\bF};\bH)-\left<\boldsymbol{\Lambda},\mathrm{i}Z_0\bB(\bar{\bF}+\bar{\mathbf{I}}_1)-(\bar{\mathbf{I}}_1-\bar{\bF})\right>\\
&-\frac{\mu}{2}\left\|\mathrm{i}Z_0\bB(\bar{\bF}+\bar{\mathbf{I}}_1)-(\bar{\mathbf{I}}_1-\bar{\bF})\right\|_F^2+\frac{\rho}{2} P(\{\bar{B}_{i,j}^{(\ell)}\}),
\end{aligned}$$
where $\mathcal{V}:=\{\boldsymbol{\gamma},\y,\mathbf{p},\bar{\bF}, \bB,\{\bar{B}_{i,j}\},\{\bar{B}_{i,j}^{(\ell)}\}\}$ collects all the optimization variables,  $\boldsymbol{\Lambda}\in\C^{(N+K)\times K}$ is the Lagrange multiplier associated with \eqref{con:BbarF2}, and $\mu>0$ is the corresponding penalty parameter. In our ADMM framework, we partition the optimization variables into four blocks: $(\boldsymbol{\gamma},\y), \mathbf{p},\bar{\bF}, (\bB,\{\bar{B}_{i,j}\},\{\bar{B}_{i,j}^{(\ell)}\})$. We next detail the update of each variable block.   To simplify the notation, the iteration index is omitted in the following derivation of the block updates. All variables not being optimized in the current subproblem are fixed at their most recently updated values.

\subsubsection{Update of $(\boldsymbol{\gamma},\y)$}  \label{update:gammay}
When $(\mathbf{p},\bar{\bF}, \bB,\{\bar{B}_{i,j}\},\{\bar{B}_{i,j}^{(\ell)}\})$ are fixed,  the variables $(\boldsymbol{\gamma},\y)$ are updated by maximizing $L(\mathcal{V})$ with respect to $(\boldsymbol{\gamma},\y)$. Following the standard FP updates, the optimal $\gamma_k$ and $y_k$ are given by
 \begin{equation}\label{update:gamma}
\gamma_k=\frac{|\h_k^H\bar{\mathbf{I}}_2\bar{\mathbf{f}}_k|^2p_k}{\sum_{j\neq k}|\h_k^H\bar{\mathbf{I}}_2\bar{\mathbf{f}}_j|^2p_j+\sigma^2},~ k=1,2,\dots, K,
\end{equation} and 
\begin{equation}\label{update:y}
y_k=\frac{\h_k^H\bar{\mathbf{I}}_2\bar{\mathbf{f}}_k\sqrt{p_k}}{\sum_{j=1}^K|\h_k^H\bar{\mathbf{I}}_2\bar{\mathbf{f}}_j|^2p_j+\sigma^2},~ k=1,2,\dots, K.
\end{equation}
\subsubsection{Update of $\mathbf{p}$}\label{update:p}
 When $(\boldsymbol{\gamma},\y, \bar{\bF}, \bB,\{\bar{B}_{i,j}\},\{\bar{B}_{i,j}^{(\ell)}\})$ are fixed, maximizing  $L(\mathcal{V})$ over $\mathbf{p}$ under the power constraint in \eqref{con:power2} yields the following problem:
\begin{equation}\label{p-subproblem}
\begin{aligned}
\max_{\mathbf{p}}~&\sum_{k=1}^K(2\alpha_k\sqrt{p_k}-\beta_kp_k),~\\
\text{s.t. }~&\mathbf{1}^T\mathbf{p}\leq P_T,~ \mathbf{p}\geq \mathbf{0},
\end{aligned}
\end{equation}
where $\alpha_k=(1+\gamma_k)\RR(y_k^*\h_k^H\bar{\mathbf{I}}_2\bar{\mathbf{f}}_k)\geq 0$ and $\beta_k=\sum_{j=1}^K(1+\gamma_j)|y_j|^2|\h_j^H\bar{\mathbf{I}}_2\bar{\mathbf{f}}_k|^2\geq0$.
The solution to \eqref{p-subproblem} can be derived from its Karush–Kuhn–Tucker (KKT) conditions. Let $\lambda$ be the Lagrange multiplier associated with $\mathbf{1}^T\mathbf{p}\leq P_T$. The KKT condition yields  
\begin{equation}\label{update:pk}
\begin{aligned}
p_k=\frac{\alpha_k^2}{(\beta_k+\lambda)^2},~~k=1,2,\dots,K,
\end{aligned}
\end{equation}
where $\lambda$ and $\mathbf{p}$ satisfy the following complementary slackness condition and  feasibility conditions:
\begin{subequations}\label{kkt:pk}
\begin{align}
&\lambda(\mathbf{1}^T\mathbf{p}-P_T)=0, \label{complementary}\\
&\mathbf{1}^T\mathbf{p}\leq P_T,~{\lambda}\geq {0}.\label{feasibility}
\end{align}
\end{subequations} 
Combining \eqref{update:pk} and \eqref{kkt:pk}, we can conclude that $\lambda=0$ if $\sum_{k=1}^K\alpha_k^2/\beta_k^2\leq P_T$. Otherwise, $\lambda$ is the unique positive solution to
 $\sum_{k=1}^K{\alpha_k^2}/{(\beta_k+\lambda)^2}=P_T.$

\subsubsection{Update of $\bar{\bF}$} \label{update:barF}
When $(\boldsymbol{\gamma},\y,\mathbf{p}, \bB,\{\bar{B}_{i,j}\},\{\bar{B}_{i,j}^{(\ell)}\})$ are fixed, maximizing  $L(\mathcal{V})$ with respect to $\bar{\bF}$ is equivalent to solving   
\begin{equation}\label{subproblemY}
\begin{aligned}
\min_{\bar{\bF}}~&\text{tr}\big(\text{diag}(\mathbf{p})(\bar{\bF})^H\bQ\bar{\bF}\big)-2\RR\big(\text{tr}(\text{diag}(\mathbf{p})^{\frac{1}{2}}\mathbf{L}\bar{\bF})\big)\\
&+\frac{\mu}{2}\left\|(\mathbf{I}+\mathrm{i}Z_0\bB)\bar{\bF}-(\mathbf{I}-\mathrm{i}Z_0\bB)\bar{\mathbf{I}}_1+{\boldsymbol{\Lambda}}/{\mu}\right\|_F^2
\end{aligned}
\end{equation}
where $\bQ=\bar{\mathbf{I}}_2^H{\bH}^H\text{diag}(({\mathbf{1}+\boldsymbol{\gamma}})\circ\y\circ({\y})^*){\bH}\bar{\mathbf{I}}_2$ and $\mathbf{L}=\text{diag}((\mathbf{1}+\boldsymbol{\gamma})\circ(\y)^*){\bH}\bar{\mathbf{I}}_2$. 
Hence, $\bar{\bF}$ is updated column-wise as 
\begin{equation}\label{update:f}
\begin{aligned}
        \bar{\mathbf{f}}_k=(p_k&\bQ+\frac{\mu}{2}(\mathbf{I}+Z_0^2\bB^2))^{-1}\big(\sqrt{p_k}[\mathbf{L}]_{k,:}^H-\frac{1}{2}(\mathbf{I}-\mathrm{i}Z_0\bB)\blam_k\\
        &+\frac{\mu}{2}(\mathbf{I}-Z_0^2{\bB}^2-\mathrm{i}2Z_0\bB)[\bar{\mathbf{I}}_{1}]_{:,k}\big),~k=1,2,\dots, K,\\
\end{aligned}
\end{equation}
where $\bar{\mathbf{f}}_k$ is the $k$-th column of $\bar{\mathbf{F}}$ and $\boldsymbol \lambda_k$ denotes the $k$-th column of $\boldsymbol \Lambda$.

\subsubsection{Update of $(\bB,\{\bar{B}_{i,j}\},\{\bar{B}_{i,j}^{(\ell)}\})$}\label{update:B}
 When $(\boldsymbol{\gamma},\y, \mathbf{p}, \bar{\bF},\boldsymbol{\Lambda})$ are fixed, maximizing  $L(\mathcal{V})$ over $(\bB,\{\bar{B}_{i,j}\},\{\bar{B}_{i,j}^{(\ell)}\})$ under the corresponding constraints in \eqref{con:BbarB2}-\eqref{con:BBl2},  \eqref{con:reciprocal2},  and \eqref{con:box} yields the following problem:
\begin{subequations}\label{updateB}
\begin{align}
\min~~&\frac{\mu}{2}\left\|\bB\bM-\boldsymbol{\Gamma}\right\|_F^2-\frac{\rho}{2}\sum_{\ell=1}^{b}\sum_{\substack{(i,j)\in\mathcal E, j\geq i}}
(\bar{B}_{i,j}^{(\ell)})^2\\
\text{s.t.}~~&[\bB]_{i,j}=\left\{
\begin{aligned}
&-\bar{B}_{i,j},~~~~~~~~~\,~~&\text{if}~i\neq j;\\
&\bar{B}_{i,i}+\textstyle\sum_{n\neq i}\bar{B}_{i,n},~~&\text{if }i=j.
\end{aligned}\right. \label{con:BbarB3} \\
&\bar{B}_{i,j}=0,~\forall~(i,j)\in\mathcal{D},\\
&\bar{B}_{i,j}= \sum_{\ell=1}^b 2^{\ell-1} \bar{B}_{i,j}^{(\ell)},~\forall~(i,j)\in\mathcal{E},\label{con:barBbarBl3}\\
~&\bar{B}_{i,j}^{(\ell)}\in[-c,c],~\bar{B}_{j,i}^{(\ell)}=\bar{B}_{i,j}^{(\ell)},~\forall~(i,j)\in \mathcal{E},~\ell=1,\dots,b, 
\end{align}
\end{subequations}
where the problem has been expressed in the real domain by defining
\begin{equation}\label{def:M}
\begin{aligned}
{\bM}&=[\RR(\mathrm{i}Z_0\bar{\bF}+\mathrm{i}Z_0\bar{\mathbf{I}}_1),~\I(\mathrm{i}Z_0\bar{\bF}+\mathrm{i}Z_0\bar{\mathbf{I}}_1)]\in\R^{(N+K)\times 2K}\\
\end{aligned}
\end{equation}
and 
\begin{equation}\label{def:Gamma}
\boldsymbol{\Gamma}=\left[\RR\left(\bar{\mathbf{I}}_1-\bar{\bF}-\frac{\boldsymbol{\Lambda}}{\mu}\right),\I\left(\bar{\mathbf{I}_1}-\bar{\bF}-\frac{\boldsymbol{\Lambda}}{\mu}\right)\right]\in\R^{(N+K)\times 2K}.
\end{equation}
Due to the symmetry and sparsity constraints, the independent tunable variables of problem \eqref{updateB} are only $\bar B_{i,j}^{(\ell)}$ with $(i,j)\in\mathcal E$, $j\ge i$, and $\ell=1,\ldots,b$. In the following, we first collect these independent variables into a compact vector, denoted by $\mathbf q$, and then express the quadratic objective function as a quadratic function of $\mathbf q$, which  transforms problem \eqref{updateB} into a box-constrained quadratic program.

For notational convenience, we use $\bar{\mathbf B}$ and $\bar{\mathbf B}^{(\ell)}$ to denote the matrix forms of $\{\bar B_{i,j}\}$ and $\{\bar B_{i,j}^{(\ell)}\}$, respectively, where entries corresponding to $(i,j)\in\mathcal D$ are set to zero. For each $i=1,2,\dots, N+K$, define 
$$\mathcal{S}_i\triangleq \{j\mid (i,j)\in\mathcal E,\ j\ge i\}=\{i=i_1<i_2\cdots<i_{|\mathcal{S}_i|}\},$$
and for each $\ell=1,2,\dots, b$, define 
\begin{equation}\label{def:ql}
\mathbf{q}^{(\ell)}=[[\bar{\bB}^{(\ell)}]_{1,\mathcal{S}_1},\dots,[\bar{\bB}^{(\ell)}]_{N+K,\mathcal{S}_{N+K}}]^T\in\R^{n_{I}\times 1},
\end{equation}
i.e.,  $\mathbf{q}^{(\ell)}$ stacks the upper-triangular entries of $\bar{\bB}^{(\ell)}$ indexed by $\mathcal E$ in a row-wise order,  where $n_{I}=\sum_{i=1}^{N+K}|\mathcal{S}_i|$.
We further stack $\{\bq^{(\ell)}\}$ as 
\begin{equation}\label{def:q}
\mathbf q=\left[(\mathbf q^{(1)})^T,(\mathbf q^{(2)})^T,\ldots,(\mathbf q^{(b)})^T\right]^T
\in\mathbb R^{bn_{I}\times 1}.
\end{equation}
The vector $\mathbf{q}$ then collects all the independent tunable variables of problem \eqref{updateB}.  Similarly,  define 
$$
\bar{\bb}=[[\bar{\bB}]_{1,\mathcal{S}_1},[\bar{\bB}]_{2,\mathcal{S}_2},\dots,[\bar{\bB}]_{N+K,\mathcal{S}_{N+K}}]^T\in\mathbb R^{n_{I}\times 1},
$$
which collects the tunable entries  in $\bar{\bB}$.  Utilizing the constraint in \eqref{con:barBbarBl3}, the vector $\bar{\bb}$ can be expressed as a linear combination of $\{\bq^{(\ell)}\}$:
\begin{equation}\label{barxbarxl}
\bar{\bb}=\sum_{\ell=1}^b2^{\ell-1}\bq^{(\ell)}.
\end{equation}

Next, we express $\|\mathbf B\mathbf M-\mathbf\Gamma\|_F^2$ as a quadratic function of  $\bar{\mathbf b}$. Let
\begin{equation}
\mathbf M^T=[\mathbf a_1,\mathbf a_2,\ldots,\mathbf a_{N+K}]
\end{equation}
and define
\begin{equation}\label{def:d}
\mathbf d=\operatorname{vec}(\mathbf\Gamma^T),
\end{equation}
where $\mathbf a_i\in\mathbb R^{2K\times 1}$ denotes the $i$-th column of $\mathbf M^T$. Since  $\mathbf B$ is linearly determined by the tunable susceptances in $\bar{\mathbf b}$ through \eqref{con:BbarB3}, the vectorized product $\operatorname{vec}((\mathbf B\mathbf M)^T)$ can be expressed as a linear function of $\bar{\mathbf b}$, i.e.,
\begin{equation}
\operatorname{vec}\!\left((\mathbf B\mathbf M)^T\right)
=\bar{\mathbf A}\bar{\mathbf b},
\end{equation}
for an appropriately constructed matrix
$\bar{\mathbf A}\in\mathbb R^{2K(N+K)\times n_I}$. We now construct $\bar{\mathbf A}$ explicitly. Define the block matrix $\bar{\mathbf A}_{i,j}\in\mathbb R^{2K\times|\mathcal S_j|}$, which characterizes the contribution of the tunable susceptances collected in $[\bar{\mathbf B}]_{j,\mathcal S_j}$ to the $i$-th row of $\mathbf B\mathbf M$, as follows:
\begin{equation}
[\bar{\mathbf A}_{i,j}]_{:,\nu}
=
\begin{cases}
\mathbf a_j, & i=j,\ \nu=1,\\
\mathbf a_j-\mathbf a_{j_\nu}, & i=j,\ \nu\ge2,\\
\mathbf a_{j_\nu}-\mathbf a_j, & i=j_\nu,\ \nu\ge2,\\
\mathbf 0, & \text{otherwise}.
\end{cases}
\end{equation}
The first case captures the contribution of the susceptance  $\bar B_{j,j}$ to the $j$-th row of $\mathbf B\mathbf M$,   whereas the remaining nonzero cases account for the contributions of an inter-port susceptance $\bar B_{j,j_\nu}$ to the $j$-th and $j_\nu$-th rows of $\mathbf B\bM$, respectively. Stacking these blocks gives
\begin{equation}
\bar{\mathbf A}
=
\begin{bmatrix}
\bar{\mathbf A}_{1,1} & \bar{\mathbf A}_{1,2} & \cdots & \bar{\mathbf A}_{1,N+K}\\
\bar{\mathbf A}_{2,1} & \bar{\mathbf A}_{2,2} & \cdots & \bar{\mathbf A}_{2,N+K}\\
\vdots & \vdots & \ddots & \vdots\\
\bar{\mathbf A}_{N+K,1} & \bar{\mathbf A}_{N+K,2} & \cdots & \bar{\mathbf A}_{N+K,N+K}
\end{bmatrix}.
\end{equation}
It then follows that
\begin{align}
\|\mathbf B\mathbf M-\mathbf\Gamma\|_F^2
&=
\left\|
\operatorname{vec}\!\left((\mathbf B\mathbf M)^T\right)
-\operatorname{vec}(\mathbf\Gamma^T)
\right\|^2 \notag\\
&=
\|\bar{\mathbf A}\bar{\mathbf b}-\mathbf d\|_2^2.
\end{align}

Using \eqref{barxbarxl}, we further have
\begin{align}
\bar{\mathbf A}\bar{\mathbf b}
&=\bar{\bA}\sum_{\ell=1}^{b}2^{\ell-1}\mathbf q^{(\ell)}=\mathbf A\mathbf q,
\end{align}
where
\begin{equation}\label{def:A}
\bA=
[\bar{\mathbf A},\,2\bar{\mathbf A},\,\ldots,\,2^{b-1}\bar{\mathbf A}]
\in\mathbb R^{2K(N+K)\times bn_I}.
\end{equation}
Therefore,
\begin{equation}
\|\mathbf B\mathbf M-\mathbf\Gamma\|_F^2
=
\|\mathbf A\mathbf q-\mathbf d\|_2^2,
\end{equation}
and problem \eqref{updateB} reduces  to

\begin{equation}\label{updatebarx}
\begin{aligned}
\min_{\mathbf q}~ &\frac{\mu}{2}\|\mathbf A\mathbf q-\mathbf d\|_2^2-\frac{\rho}{2}\|\mathbf q\|_2^2 \\
\text{s.t.}~~ &
[\mathbf q]_i\in[-c,c],
\quad i=1,2,\ldots,bn_{I}.
\end{aligned}
\end{equation}
Hence, the original $\mathbf B$-related update is reduced to a box-constrained quadratic program over the  two-level variables collected in $\mathbf q$.

We solve it using the projected gradient descent (PGD) method. Specifically, at iteration $r$, the variable $\mathbf{q}$ is updated according to
\begin{equation}\label{updateq}
\mathbf{q}^{[r+1]}=\Pi_{[-c,c]}\bigl(\mathbf{q}^{[r]}-\alpha\mathbf{g}^{[r]}\bigr),
\end{equation}
where
\begin{equation}
\mathbf{g}^{[r]}=\mu\mathbf{A}^{T}(\mathbf{A}\mathbf{q}^{[r]}-\mathbf{d})-\rho\mathbf{q}^{[r]}
\end{equation}
is the gradient of the objective function evaluated at $\mathbf{q}^{[r]}$, $\alpha>0$ denotes the stepsize, and $\Pi_{[-c,c]}(\cdot)$ denotes the Euclidean projection onto the box set $[-c,c]$. The projection is performed element-wise as
\begin{equation}
[\Pi_{[-c,c]}(\mathbf{z})]_i
=
\operatorname{sgn}(z_i)
\min\{|z_i|,c\}.
\end{equation}
After obtaining $\mathbf q$, we recover $\{\bar{\mathbf B}^{(\ell)}\}_{\ell=1}^{b}$ from
$\{\mathbf q^{(\ell)}\}_{\ell=1}^{b}$, construct
$\bar{\mathbf B}$ from \eqref{con:barBbarBl3}, and then obtain $\mathbf B$ according to \eqref{con:BbarB3}.

\subsubsection{Update of $\boldsymbol \Lambda$} Finally, the Lagrange multiplier is updated as 
\begin{equation}\label{updatelambda}
\boldsymbol{\Lambda}=\boldsymbol{\Lambda}+\mu(\mathrm{i}Z_0\bB(\bar{\bF}+\bar{\mathbf{I}}_1)-(\bar{\mathbf{I}}_1-\bar{\bF})).
\end{equation}
\begin{algorithm}[t]
\caption{ADMM Algorithm for Solving Problem \eqref{problem2:sumrate-penalty}}
\label{alg:online_admm}
\begin{algorithmic}[1]
\State \textbf{Input:} Channel matrix $\mathbf H$, quantization parameter $c$, penalty parameters $\rho$ and $\mu$, maximum transmit power $P_T$.
\State \textbf{Output:} Beamforming matrix $\bar{\mathbf F}$, power allocation vector $\mathbf p$, and susceptance matrix $\mathbf B$.
\State Initialize $\mathbf p^{0}$, $\bar{\mathbf F}^{0}$,  $\mathbf B^{0}$, and $\boldsymbol{\Lambda}^{0}$.
\State Set $t=0$.
\Repeat
    \State Update $(\boldsymbol{\gamma}^{t+1},\mathbf y^{t+1})$ by \eqref{update:gamma} and \eqref{update:y}. 
    \State Update $\mathbf p^{t+1}$ by \eqref{update:pk}.
        \State Update $\bar{\mathbf F}^{t+1}$ by \eqref{update:f}. 
            \State Construct $\mathbf M^{t+1}$, $\boldsymbol{\Gamma}^{t+1}$, $\mathbf d^{t+1}$, and $\mathbf A^{t+1}$ as in \eqref{def:M}, \eqref{def:Gamma}, \eqref{def:d}, and \eqref{def:A}, respectively.
    \State Initialize the inner PGD iteration with $\mathbf q^{[0]}=\mathbf q^{t}$.
\State Set $r=0$.
    \Repeat
        \State Update $\bq^{[r+1]}$ by \eqref{updateq}. 
        \State Set $r=r+1$.
    \Until{convergence}
    \State Set $\bq^{t+1}=\bq^{[r]}.$
    \State Recover $\{\bar{\mathbf B}^{(\ell),t+1}\}_{\ell=1}^{b}$ from $\bq^{t+1}$, and construct $\bar{\bB}^{t+1}$ and $\bB^{t+1}$ according to  \eqref{con:barBbarBl3} and  \eqref{con:BbarB3}, respectively.
    \State Update the Lagrange multiplier by \eqref{updatelambda}.
    \State Set $t=t+1$.
\Until{convergence}
\end{algorithmic}
\end{algorithm}

A summary of the above ADMM procedure is given in Algorithm \ref{alg:online_admm}.

\begin{remark}
A very recent work \cite{zhang2026beamforming} also studied finite-resolution MiLAC design. The methods in \cite{zhang2026beamforming} first solve a continuous susceptance problem and then project the obtained solution onto the discrete feasible set. To reduce the resulting quantization loss, \cite{zhang2026beamforming} further proposes an alternating refinement method, which updates one susceptance at a time. This approach provides an efficient post-hoc refinement strategy, but its performance can be sensitive to the quality of the initial point due to the discrete and nonconvex nature of the problem. In contrast, the proposed approach is based on an exact continuous penalty model. The continuation framework in Algorithm \ref{alg:rho_update} first seeks a favorable continuous configuration under a moderate penalty and then uses the solution from each stage to initialize the next stage with a larger penalty. This allows the algorithm to gradually enforce discreteness while retaining useful information from the preceding solutions, which may help reduce sensitivity to the initial configuration. Moreover, our framework includes an offline stage for optimizing the quantization parameter $c$ based on channel statistics (see Section \ref{sec:offline}), rather than assuming a fixed hardware dynamic range and grid as in \cite{zhang2026beamforming}. Numerical comparisons with the baselines in \cite{zhang2026beamforming} are provided in Section \ref{sec:simulation}.
\end{remark}

\section{Algorithm for Solving the Offline Codebook Design Problem}\label{sec:offline}
In this section, we consider the offline codebook design problem in \eqref{problem:codebook}, which aims to optimize the quantization parameter $c$ based on the channel statistics. Directly solving \eqref{problem:codebook} is challenging because the expectation cannot be computed analytically, as $R^\star(c;\bH)$ is the optimal value of a non-convex online discrete beamforming problem involving both discrete and continuous variables, which is difficult to evaluate even for a fixed $c>0$ and a given channel realization $\bH$.
To address this issue, we adopt a sample average approximation (SAA)-based method. The resulting SAA problem is first formulated in Section \ref{subsec:SAA}. Following the penalty-based approach developed in Section \ref{subsec:penalty}, we then introduce a penalty model for the SAA problem in Section \ref{subsec:SAA_penalty}. Finally, Section \ref{subsec:SAA_ADMM} develops an ADMM-based algorithm for solving the resulting penalty model. 
\subsection{Sample Average Approximation}\label{subsec:SAA}
Let $\{\mathbf H^{(m)}\}_{m=1}^{M}$ denote $M$ independent channel samples drawn from the channel distribution. Then, the expectation in \eqref{problem:codebook} can be approximated by the empirical average
\begin{equation}\label{sample_problem}
\frac{1}{M}\sum_{m=1}^M R^\star(c;\bH^{(m)}),
\end{equation}
where $R^\star(c;\bH^{(m)})$ is the optimal value of the online problem in \eqref{problem2:sumrate} under the channel realization $\mathbf H^{(m)}$.
   For compactness, let $\mathcal{U}=(\mathbf{p},\bar{\mathbf F}, \bB,\{\bar{B}_{i,j}\},\{\bar{B}_{i,j}^{(\ell)}\}_{\ell=1}^{b})$ denote the collection of optimization variables in \eqref{problem2:sumrate}, and let \(\mathcal X\) denote the feasible set defined by all constraints in \eqref{problem2:sumrate} except the two-level discrete constraint in \eqref{con:binary}, i.e., $\mathcal{X}$ collects all constraints that are independent of the specific value of the quantization parameter $c$, while the dependence on $c$ is isolated in the two-level discrete constraint.

The sample average approximation problem is then formulated as
\begin{subequations}\label{problem:codebook2}
\begin{align}
\max_{c\ge 0,\{\mathcal{U}^{(m)}\}}~& \frac{1}{M}\sum_{m=1}^{M}
R\!\left(\mathbf p^{(m)},\bar{\mathbf F}^{(m)};\mathbf H^{(m)}\right)\\
\mathrm{s.t.}~~~~~&\mathcal{U}^{(m)}\in{\cal{X}},\quad m=1,\ldots,M,\\
&
\bar B_{i,j}^{(m,\ell)}\in\{-c,c\},~ \forall ~(i,j)\in\mathcal E,\\
&\hspace{1.4cm} \forall~\ell=1,\ldots,b,\ m=1,\ldots,M,\label{codebook:discrete}
\end{align}
\end{subequations}
where the superscript $m$ indicates that the corresponding variables are associated with the $m$-th channel sample.

\subsection{Penalty Model of Problem \eqref{problem:codebook2}}\label{subsec:SAA_penalty}
Following the penalty-based approach in Section  \ref{subsec:penalty}, we relax the two-level discrete constraints to their convex hull, i.e.,  $\bar{B}_{i,j}^{(m,\ell)}\in[-c,c]$, and introduce  the following quadratic penalty term to promote discreteness: 
$$P(\{\bar B_{i,j}^{(m,\ell)}\},c)=\frac{1}{M}\sum_{m=1}^M\sum_{\ell=1}^b\sum_{(i,j)\in\mathcal{E},  j\geq i}((\bar{B}_{i,j}^{(m,\ell)})^2-c^2).$$
The penalty problem is then given by 
 \begin{subequations}\label{problem:codebook2_penalty}
\begin{align}
\max_{c\ge 0,\{\mathcal{U}^{(m)}\}}\,& \frac{1}{M}\sum_{m=1}^{M}
R(\mathbf p^{(m)},\bar{\mathbf F}^{(m)};\mathbf H^{(m)})+\frac{\rho}{2}P(\{\bar B_{i,j}^{(m,\ell)}\},c)\\
\mathrm{s.t.}~~~~&\mathcal{U}^{(m)}\in{\cal{X}},\quad m=1,\ldots,M,\\
&
\bar B_{i,j}^{(m,\ell)}\in[-c,c],~ \forall ~(i,j)\in\mathcal E,\notag\\
&\hspace{1.6cm} \forall~\ell=1,\ldots,b,\ m=1,\ldots,M. \label{con:box2}
\end{align}
\end{subequations}
Unlike the online problem in Section \ref{sec:online}, the quantization parameter $c$ is now an optimization variable rather than a fixed constant. As a result, both the relaxed box constraints and the penalty term depend on $c$,  making the exactness analysis more involved.  The following theorem shows that the exactness property still holds when the penalty parameter is sufficiently large.

 \begin{theorem}\label{equivalence2}
There exists a constant $\bar{\rho}>0$ such that, for any $\rho>\bar{\rho}$, problems  \eqref{problem:codebook2} and \eqref{problem:codebook2_penalty} are globally equivalent. Specifically, any optimal solution of  \eqref{problem:codebook2} is also optimal to \eqref{problem:codebook2_penalty}. 
  Conversely, any optimal solution to problem \eqref{problem:codebook2_penalty} satisfies the two-level discrete constraint in \eqref{codebook:discrete}, and is further an optimal solution to problem \eqref{problem:codebook2}.
\end{theorem}
\begin{proof}
See Appendix \ref{app:equivalence2}.
\end{proof}

\subsection{Algorithm for Solving the Penalty Model in \eqref{problem:codebook2_penalty}}\label{subsec:SAA_ADMM}
In this subsection, we develop an ADMM-based algorithm for solving the penalty problem in \eqref{problem:codebook2_penalty}. Compared with the online penalty problem in \eqref{problem2:sumrate-penalty}, problem \eqref{problem:codebook2_penalty} contains $M$ sample-dependent copies of the online variables, while the quantization parameter $c$ is shared
by all channel samples. Therefore, the updates of most variables can be performed sample-wise and follow the same form as those in Algorithm \ref{alg:online_admm}.

Specifically, for each channel sample $m$, we introduce the FP auxiliary
variables $(\boldsymbol{\gamma}^{(m)},\y^{(m)})$ and the dual
variable $\boldsymbol{\Lambda}^{(m)}$ associated with the bilinear
constraint
\[\mathrm{i}Z_0\mathbf B^{(m)}(\bar{\mathbf F}^{(m)}+\bar{\mathbf I}_1)=\bar{\mathbf I}_1-\bar{\mathbf F}^{(m)}.\]
Given $\bB$-related variables, i.e., $\{\bB^{(m)}\},\{\bar{B}_{i,j}^{(m)}\},\{\bar{B}_{i,j}^{(m,\ell)}\}$,  and $c$, the updates of $(\boldsymbol{\gamma}^{(m)},\y^{(m)})$,
$\mathbf{p}^{(m)}$, and $\bar{\bF}^{(m)}$ are identical to those
in Algorithm \ref{alg:online_admm}, with
\(\mathbf H\), \(\mathbf p\), \(\bar{\mathbf F}\), \(\mathbf B\), and
\(\boldsymbol{\Lambda}\) replaced by
\(\mathbf H^{(m)}\), \(\mathbf p^{(m)}\), \(\bar{\mathbf F}^{(m)}\),
\(\mathbf B^{(m)}\), and \(\boldsymbol{\Lambda}^{(m)}\), respectively.
The details are omitted to avoid repetition.

The main difference from Algorithm \ref{alg:online_admm} lies in the update of the $\bB$-related variables. In the offline problem, the quantization parameter $c$ is also an optimization variable and is shared across all channel samples. Hence, the $\bB$-related variables from different samples are
coupled through the common box constraints in \eqref{con:box2}.

Following the vectorization in Section \ref{update:B}, let $\mathbf q^{(m)}\in\mathbb R^{bn_I\times 1}$ collect all upper-triangular two-level variables associated with the $m$-th channel sample. For fixed
\(\{(\boldsymbol{\gamma}^{(m)},\mathbf y^{(m)},\mathbf p^{(m)},
\bar{\mathbf F}^{(m)},\boldsymbol{\Lambda}^{(m)})\}_{m=1}^{M}\), the
joint update of \(\{\mathbf q^{(m)}\}_{m=1}^{M}\) and \(c\) can be written as
\begin{equation}\label{qsubproblem}
\begin{aligned}
\min_{\{\mathbf q^{(m)}\},c\ge 0}~&
\frac{1}{M}\sum_{m=1}^{M}
\left(\frac{\mu}{2}\|\mathbf A^{(m)}\mathbf q^{(m)}-\mathbf d^{(m)}\|_2^2-\frac{\rho}{2}\|
\mathbf q^{(m)}\|_2^2\right)\\
&+\frac{\rho}{2}bn_Ic^2
\\
\mathrm{s.t.}\quad
&
[\mathbf q^{(m)}]_i\in[-c, c],~i=1,\ldots,bn_I,~
m=1,\ldots,M .
\end{aligned}
\end{equation}
Here, $\bd^{(m)}$ and $\bA^{(m)}$ are constructed in the same way as $\bd$ and $\bA$ in \eqref{def:d} and \eqref{def:A}, respectively, using the variables associated with the $m$-th channel sample. When $c$ is fixed, problem \eqref{qsubproblem} reduces to \(M\) independent online $\bB$-subproblems.
When \(c\) is optimized, however, all samples are coupled through the
common constraint $[\mathbf q^{(m)}]_i\in[-c,c]$.

We next rewrite problem \eqref{qsubproblem}. Let
$$\bq=[(\mathbf q^{(1)})^T,\ldots,(\mathbf q^{(M)})^T]^T\in\R^{Mbn_I\times 1}.$$
For fixed $\bq$, the optimal $c$ is the smallest value satisfying all box constraints, i.e.,
$c^\star=\|\mathbf q\|_\infty.$
Therefore, after eliminating $c$ and multiplying the objective by $M$, problem \eqref{qsubproblem} can be equivalently written as the following unconstrained optimization problem over $\bq\in\R^{Mbn_I\times 1}$:
\begin{equation}\label{qsubproblem2}
\min_{\mathbf{q}}
\frac{\mu}{2}\|\bA\bq-\bd\|_2^2-\frac{\rho}{2}\|\mathbf{q}\|_2^2+\frac{\rho}{2}Mbn_I\|\mathbf q\|_\infty^2,
\end{equation}
where $\bd=[(\mathbf d^{(1)})^T,\ldots,(\mathbf d^{(M)})^T]^T$ and $\bA=\operatorname{blkdiag}(\bA^{(1)},\dots,\bA^{(M)}).$

The above problem consists of a smooth quadratic term and a nonsmooth regularization term involving $\|\bq\|_\infty^2$.  We solve it by the proximal gradient method \cite{parikh2014proximal}.  Define
$$f(\mathbf q)=\frac{\mu}{2}\|\mathbf A\mathbf q-\mathbf d\|_2^2-\frac{\rho}{2}\|\mathbf q\|_2^2,~
h(\mathbf q)=\frac{\rho}{2}Mbn_I\|\mathbf q\|_\infty^2.$$
At the $r$-th inner iteration, $\bq$ is updated as
$$\mathbf q^{[r+1]}=\operatorname{prox}_{\alpha h}\left(\mathbf q^{[r]}-\alpha\nabla f(\mathbf q^{[r]})\right),
$$
where $\alpha>0$ is the stepsize and
$$\nabla f(\mathbf q)=
\mu\mathbf A^T(\mathbf A\mathbf q-\mathbf d)-\rho\mathbf q.
$$
The proximal operator is given by
\[
\operatorname{prox}_{\alpha h}(\mathbf z)
=
\arg\min_{\mathbf x}
\left\{
\frac{1}{2}\|\mathbf x-\mathbf z\|_2^2
+
\frac{\alpha\rho}{2}Mbn_I\|\mathbf x\|_\infty^2
\right\}.
\]
As shown in \cite[Algorithm 1]{jacobsson2017quantized}, the proximal operator of the squared $\ell_\infty$-norm, i.e.,
$$\operatorname{prox}_{\lambda\|\cdot\|_{\infty}^2}(\mathbf{z})=\arg\min_{\x}\left\{\|\x\|_{\infty}^2+\frac{1}{2\lambda}\|\x-\mathbf{z}\|_2^2\right\},$$ can be computed very efficiently via a sorting-based closed-form procedure.   After the inner proximal gradient iteration converges, we recover
$\{\mathbf q^{(m)}\}_{m=1}^M$ from $\bq$. Then for each channel sample $m$, we  construct
\(\{\bar{\mathbf B}^{(m,\ell)}\}_{\ell=1}^{b}\),
\(\bar{\mathbf B}^{(m)}\), and \(\mathbf B^{(m)}\) as in Section \ref{update:B}. The quantization parameter is recovered as $c=\|\mathbf q\|_\infty.$

Finally, the dual variables are updated by
$$
\begin{aligned}
\boldsymbol{\Lambda}^{(m)}=\boldsymbol{\Lambda}^{(m)}+\mu(\mathrm{i}Z_0\bB^{(m)}(\bar{\bF}^{(m)}+\bar{\mathbf{I}}_1)&-(\bar{\mathbf{I}}_1-\bar{\bF}^{(m)})),\\&~\forall~m=1,\dots, M.
\end{aligned}
$$
The resulting ADMM algorithm alternates among the sample-wise updates of
\((\boldsymbol{\gamma}^{(m)},\mathbf y^{(m)})\), \(\mathbf p^{(m)}\), and
\(\bar{\mathbf F}^{(m)}\), the joint update of
$\bB$-related variables and $c$, and the dual  update.

\section{Simulation Results}\label{sec:simulation}
For all the simulations, the quantization parameter $c$ is determined by the proposed offline codebook design algorithm in Section \ref{sec:offline} using $M=20$ training channel samples. The obtained codebook is then fixed and the sum-rate performance is evaluated over another $100$ independently generated test channel samples using the proposed online design algorithm in Section \ref{sec:online}.  This section is organized into two parts. In Section \ref{subsec:simulation1}, we validate the effectiveness of the proposed algorithms. Then, in Section \ref{subsec:simulation2}, we investigate the impact of different quantization resolutions and MiLAC architectures on the sum-rate performance.
\subsection{Effectiveness of the Proposed Algorithms}\label{subsec:simulation1}
In this subsection, we evaluate the effectiveness of the proposed approaches. We focus on  Rayleigh fading channels, i.e., $\operatorname{vec}(\bH)\sim\mathcal{CN}(\mathbf{0},\mathbf{I})$, and consider a fully-connected MiLAC. The number of transmit antennas and the number of users are fixed as $N=64$ and $K=4$, respectively. The SNR is set to $10$ dB.

\begin{figure}
\includegraphics[width=0.37\textwidth]{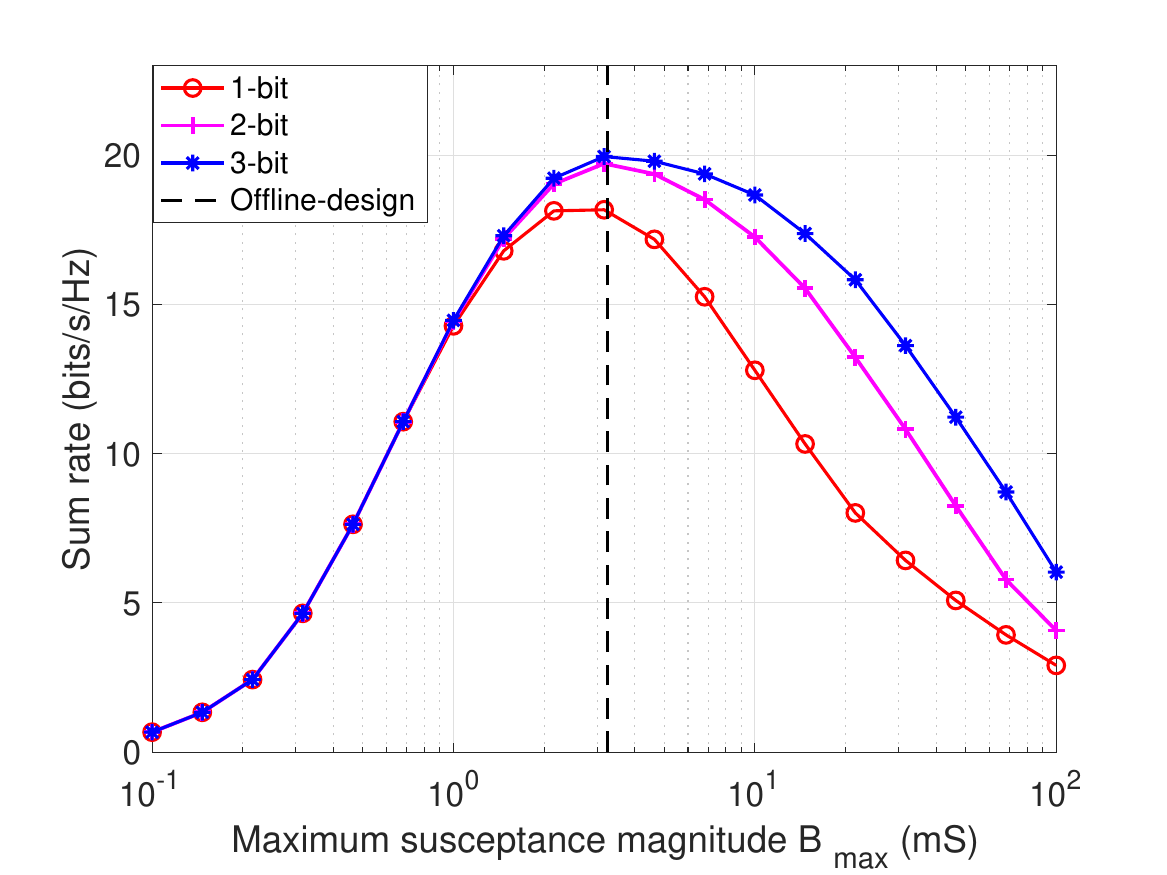}
\centering
\caption{Sum rate versus the maximum susceptance magnitude $B_{\max}$ for fully-connected MiLAC, where $N=64$, $K=4$, and $\mathrm{SNR}=10$ dB. The dashed vertical line denotes the $B_{\max}$ obtained by the proposed offline codebook-design algorithm. Since the offline-designed $B_{\max}$ values for 1-bit, 2-bit, and 3-bit codebooks are almost identical, only one dashed line is shown for clarity.}
\label{fig:c}
\end{figure}
Fig.~\ref{fig:c} evaluates the effectiveness of the proposed offline codebook-design approach in Section \ref{sec:offline}. The quantization parameter $c$ determines both the spacing between adjacent codebook levels and the maximum susceptance magnitude, which is given by $B_{\max}=(2^b-1)c$. As a benchmark, we perform an exhaustive grid search over $B_{\max}$, following  the approach in \cite{zhang2026beamforming}. For each candidate $B_{\max}$, the corresponding quantization parameter is set as $c=B_{\max}/(2^b-1)$. The offline-designed $B_{\max}$ values for 1-bit, 2-bit, and 3-bit codebooks obtained using the algorithm in Section \ref{sec:offline} are nearly identical.  Therefore, only one dashed vertical line is shown in Fig. \ref{fig:c} for clarity. This observation suggests that the proposed offline codebook-design algorithm identifies a nearly common optimal hardware dynamic range across different quantization resolutions.

As shown in Fig. \ref{fig:c}, the sum-rate first increases with $B_{\max}$ and then decreases. This is because a small dynamic range restricts the feasible susceptance values, whereas a large dynamic range leads to a coarse finite-resolution codebook. By comparing the 1-bit, 2-bit, and 3-bit cases, we observe that their performance is similar in the small-$B_{\max}$ regime, where the loss is dominated by the limited dynamic range. In contrast, at large $B_{\max}$, higher-resolution codebooks clearly outperform lower-resolution ones, since the performance loss is mainly caused by coarse quantization. Moreover, the $B_{\max}$ obtained by the proposed offline codebook-design algorithm is close to the exhaustive-search optimum, demonstrating its effectiveness as a systematic and scalable alternative to exhaustive grid search.

Fig. \ref{fig:algorithm} compares the proposed online beamforming design algorithm in Section \ref{sec:online} with two baselines in \cite{zhang2026beamforming}. Specifically, PHP refers to the post-hoc projection method in \cite[Algorithm 2]{zhang2026beamforming}, and AR refers to the alternating refinement method in \cite[Algorithm 3]{zhang2026beamforming}. For a fair comparison, all finite-resolution methods use the same codebook obtained by the proposed offline codebook-design algorithm. We also include the performance of the unquantized MiLAC obtained by the algorithm in \cite{wu2026microwave} as an upper bound.

As shown in Fig. \ref{fig:algorithm}, PHP suffers from a large performance loss because direct projection does not optimize the discrete susceptances with respect to the sum-rate objective. AR significantly improves PHP by refining the projected solution. However, since AR updates one susceptance at a time over the finite codebook, it may be trapped in a poor local solution. As a result, its performance gradually saturates as the number of quantization bits increases, leaving a noticeable gap from the unquantized benchmark.  In contrast, the proposed algorithm handles the finite-resolution constraints through a penalty-continuation framework and consistently achieves the best performance among all finite-resolution methods. In particular, as the number of quantization bits increases, the performance of the proposed method approaches that of the unquantized benchmark. 
\begin{figure}
\includegraphics[width=0.37\textwidth]{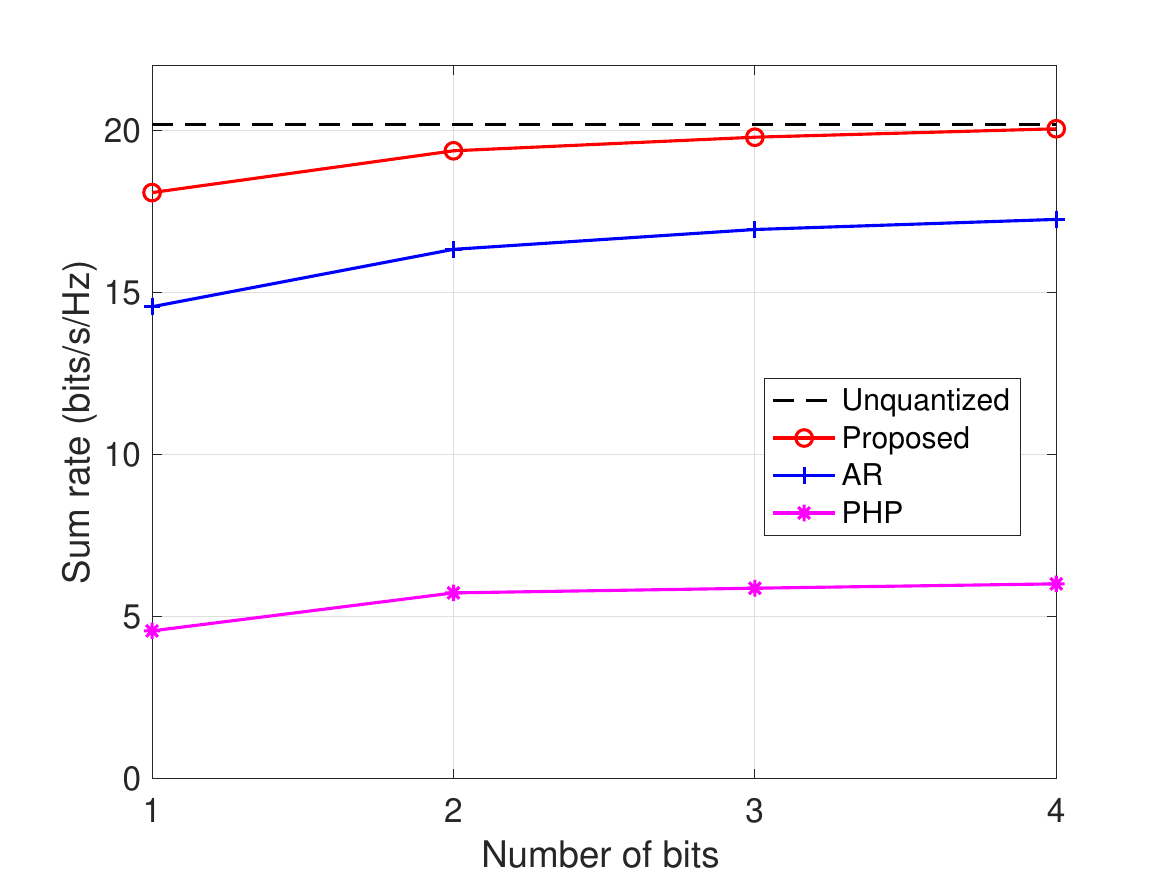}
\centering
\caption{Sum-rate versus the number of bits for fully-connected MiLAC, where $N=64$, $K=4$, and $\mathrm{SNR}=10$ dB. }
\label{fig:algorithm}
\end{figure}

\subsection{Performance Evaluation under Different Quantization Resolutions and Architectures}\label{subsec:simulation2}
In this subsection, we investigate the impact of quantization resolution and MiLAC architecture on the sum-rate performance using the proposed approach. 
\begin{figure}
\centering
\subfigure[Rayleigh fading]{\includegraphics[width=0.37\textwidth]{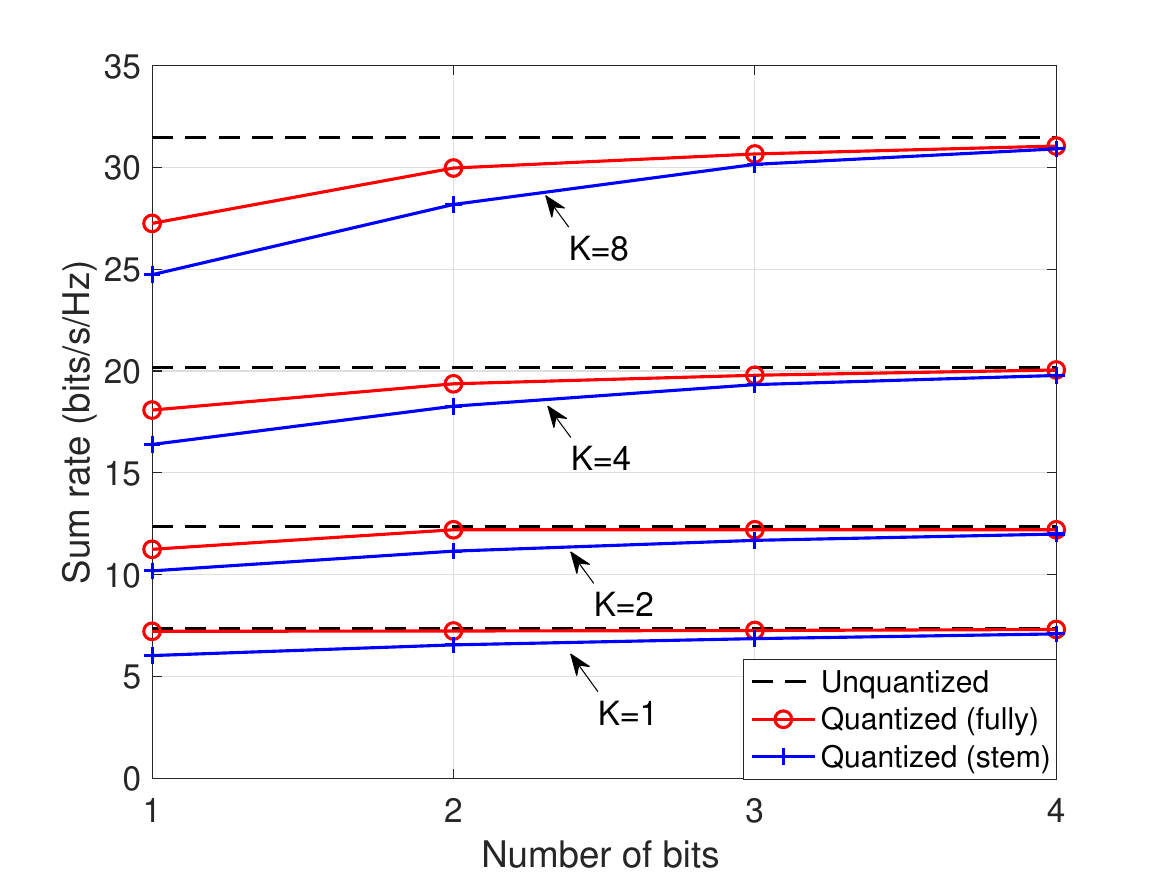}}
\subfigure[Rician fading with Rician factor $\kappa=3$.]{\includegraphics[width=0.37\textwidth]{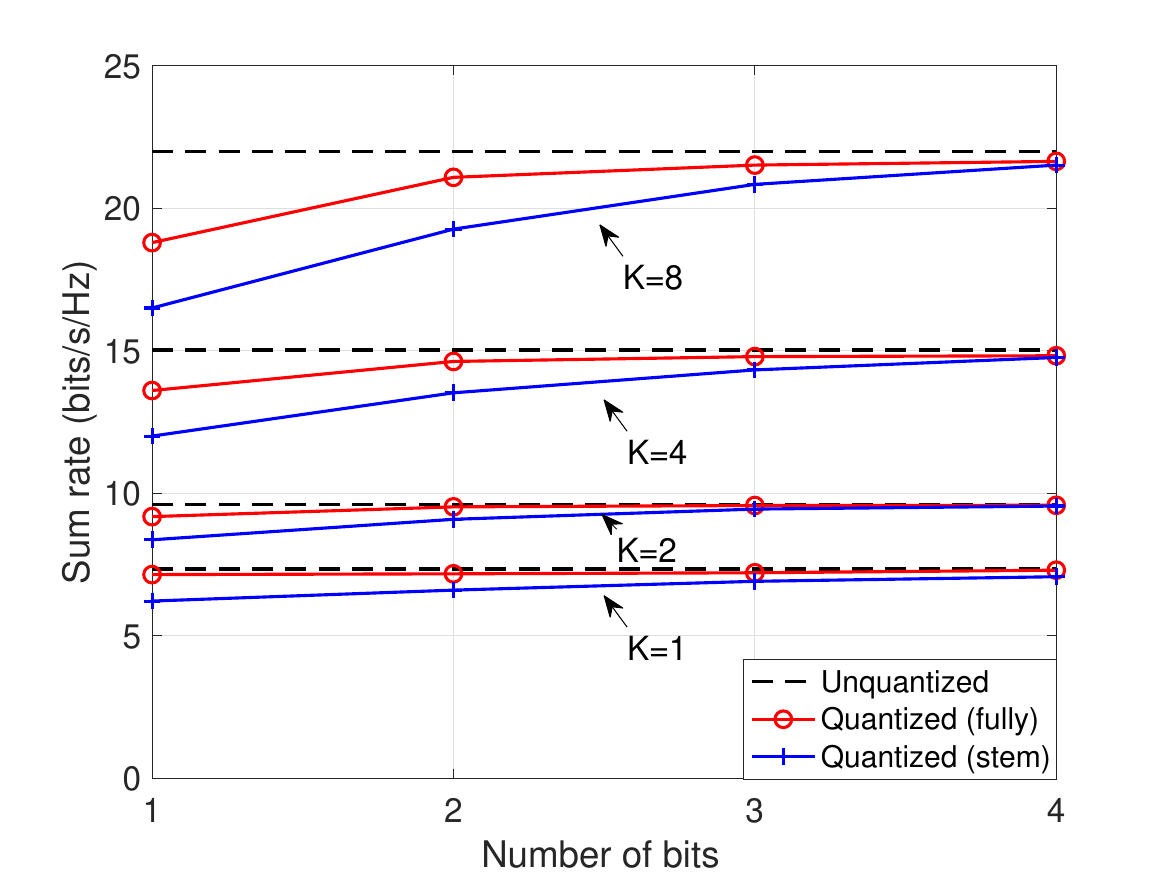}}
\caption{Sum-rate versus the number of quantization bits for fully-connected and stem-connected MiLAC, where $N=64$ and SNR$=10$ dB. Results are shown for different numbers of users under (a) Rayleigh fading and (b) Rician fading with $\kappa=3$. For the stem-connected MiLAC, the number of central ports is set as $2K-1$.}
\label{fig:bit}
\end{figure}

\begin{figure*}
\centering
\subfigure[1-bit]{\includegraphics[width=0.33\textwidth]{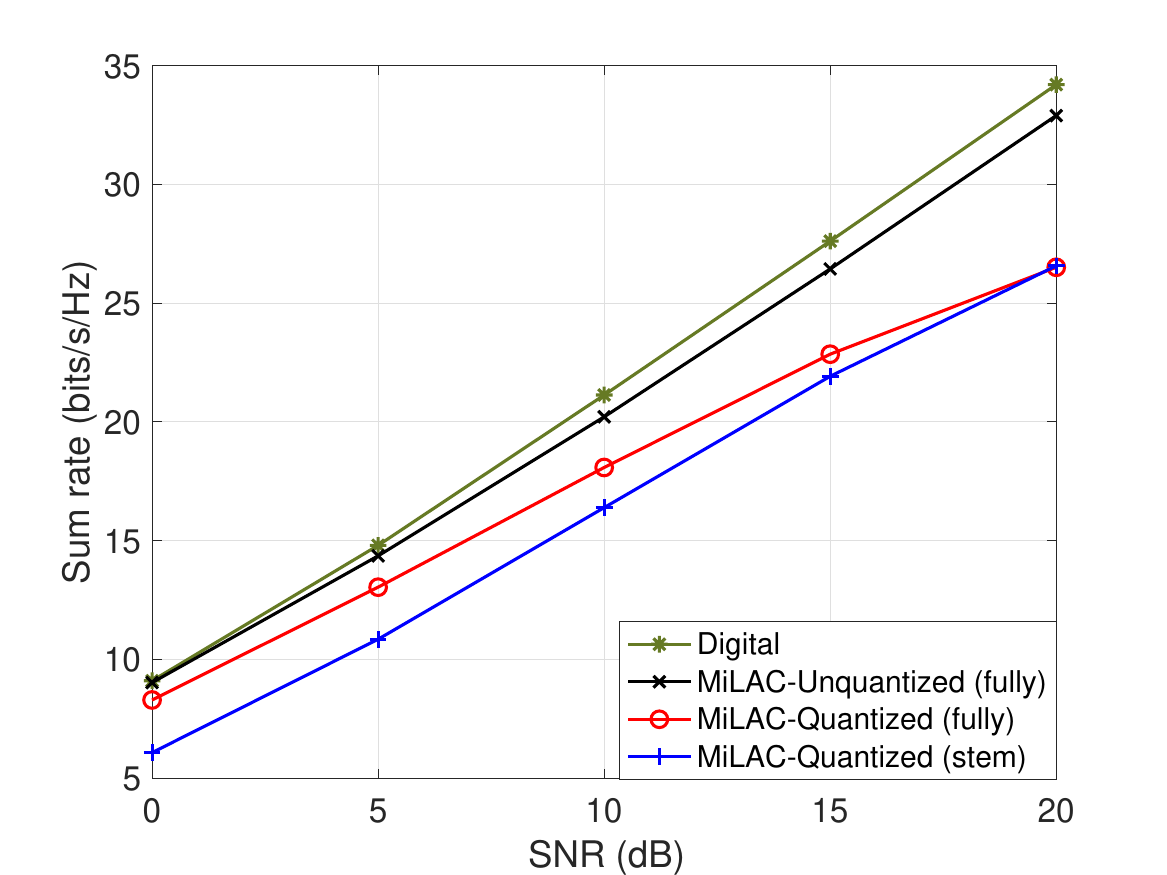}}\hspace{-0.5cm}
\subfigure[2-bit]{\includegraphics[width=0.33\textwidth]{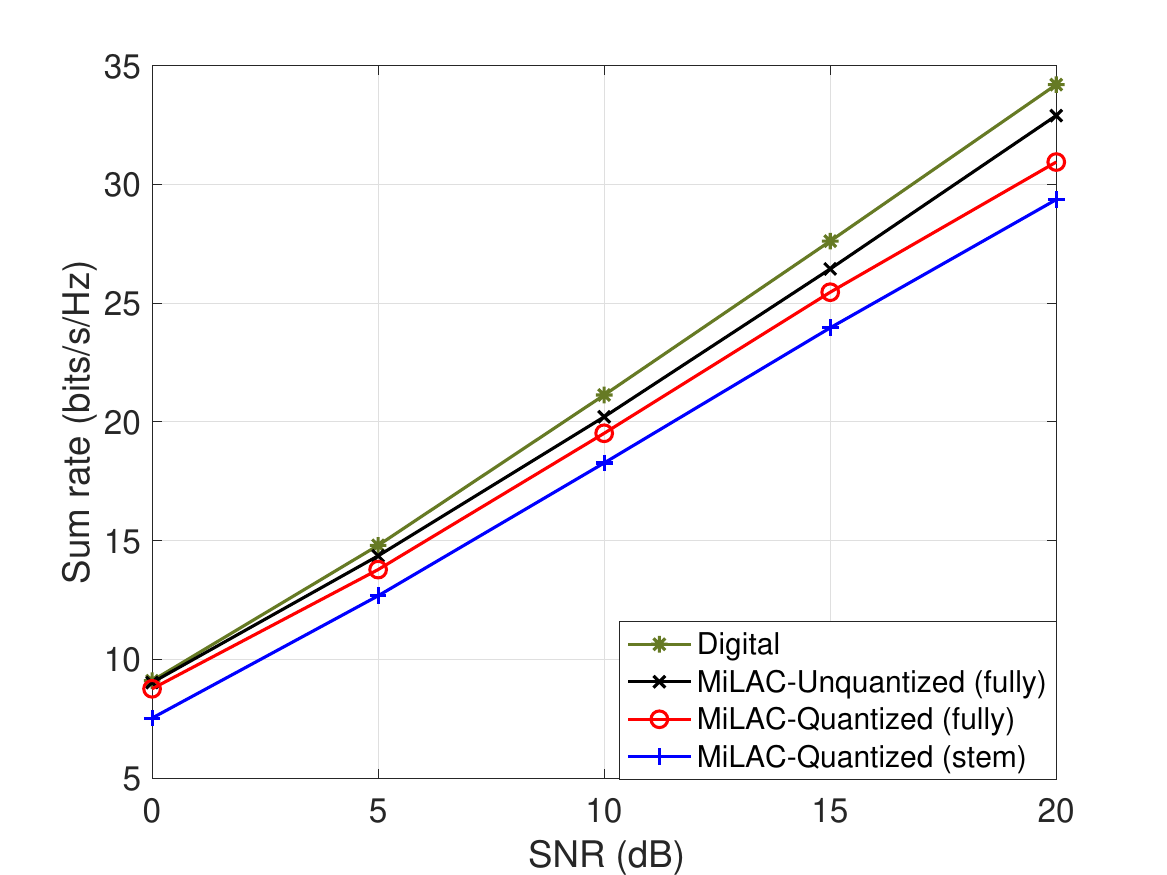}}
\subfigure[3-bit]{\includegraphics[width=0.33\textwidth]{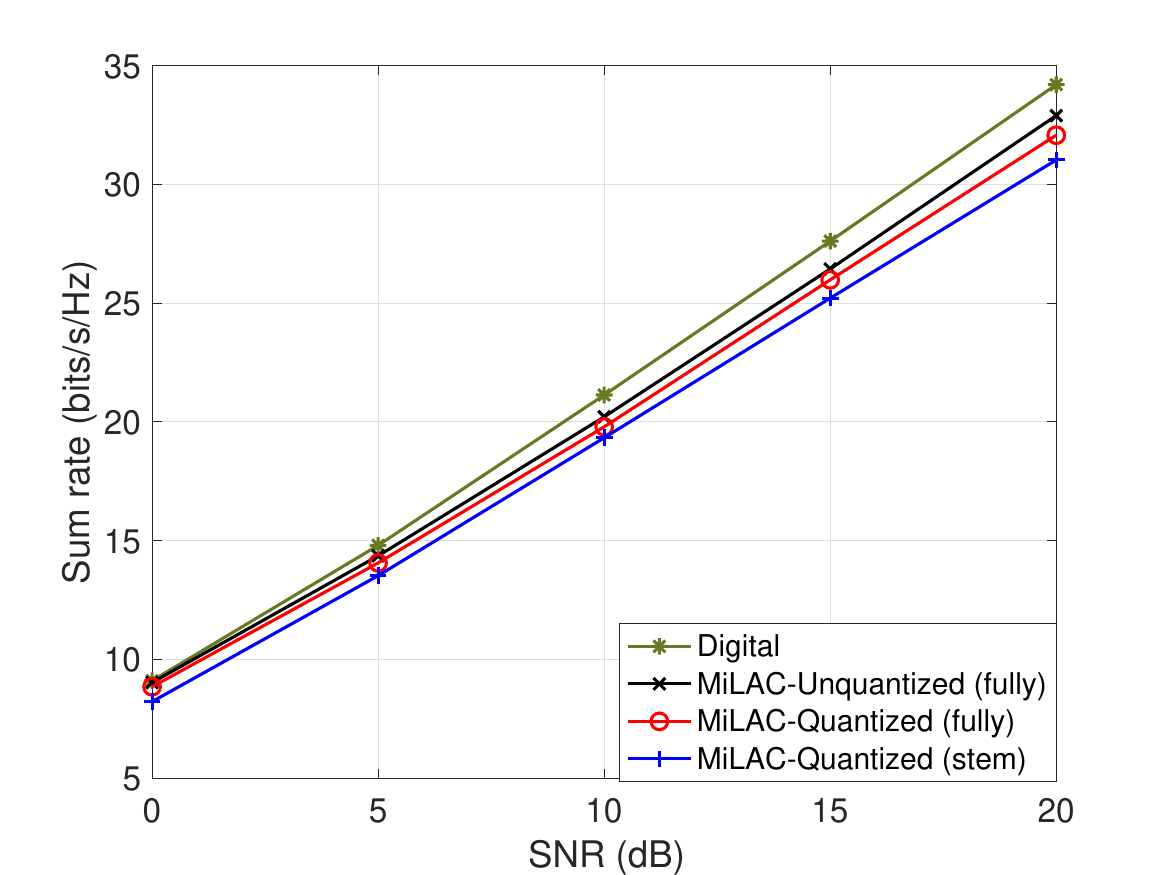}}
\caption{Sum-rate versus SNR for fully-connected and stem-connected MiLAC over Rayleigh fading channel, where $N=64$ and $K=4$. For the stem-connected MiLAC, the number of central ports is set to $2K-1$. }
\label{fig:SNR}
\end{figure*}

Fig. \ref{fig:bit} investigates the impact of quantization resolution under different numbers of users and channel distributions. Both fully-connected MiLAC and stem-connected MiLAC are considered. For the stem-connected architecture, the number of central ports is set to $2K-1$, which is sufficient to attain the fully-connected performance in the unquantized case \cite{MIMOcapacity2,zhang2026beamforming}. As observed, the sum-rate performance increases with the number of quantization bits and approaches the unquantized benchmark with only a few bits. For example, for fully-connected MiLAC, 1 bit is sufficient to approach the unquantized benchmark when $K=1$, while 2 bits are sufficient when $K=2$. As the number of users increases, more quantization bits are required because greater design flexibility is needed to manage multiuser interference (MUI). The figure also reveals the tradeoff between hardware connectivity and quantization resolution. The fully-connected architecture generally outperforms the stem-connected architecture because it provides more tunable interconnections. However, the performance gap between the two architectures decreases as the number of bits increases. In all considered cases, they achieve nearly identical performance with 4-bit quantization.  This indicates that the reduced connectivity of the stem-connected architecture can be compensated by using a higher-resolution codebook. Similar trends are observed under both Rayleigh and Rician fading channels.

Fig. \ref{fig:SNR} depicts the sum-rate performance of fully- and stem-connected MiLAC versus SNR under different quantization resolutions. Digital beamforming and unquantized fully-connected MiLAC-aided beamforming are included as performance benchmarks. As shown in Fig. \ref{fig:SNR}, the performance of unquantized fully-connected MiLAC-aided beamforming closely matches that of the digital beamforming benchmark, which is consistent with the observations in \cite{wu2026microwave}. For finite-resolution MiLAC, the gap from the unquantized benchmark becomes larger as the SNR increases, especially in the 1-bit case. This is because, in the high-SNR regime, the sum-rate is mainly limited by MUI. With only 1-bit resolution, the codebook does not provide sufficient flexibility to accurately suppress MUI. At SNR= $20$\,dB, the fully- and stem-connected MiLAC architectures even achieve similar performance in the 1-bit case. 
This suggests that the coarse quantization limits the benefit of the additional interconnections for MUI suppression.  Nevertheless, with a 3-bit codebook, both architectures approach the unquantized benchmark over the considered SNR range.


\begin{figure}
\includegraphics[width=0.37\textwidth]{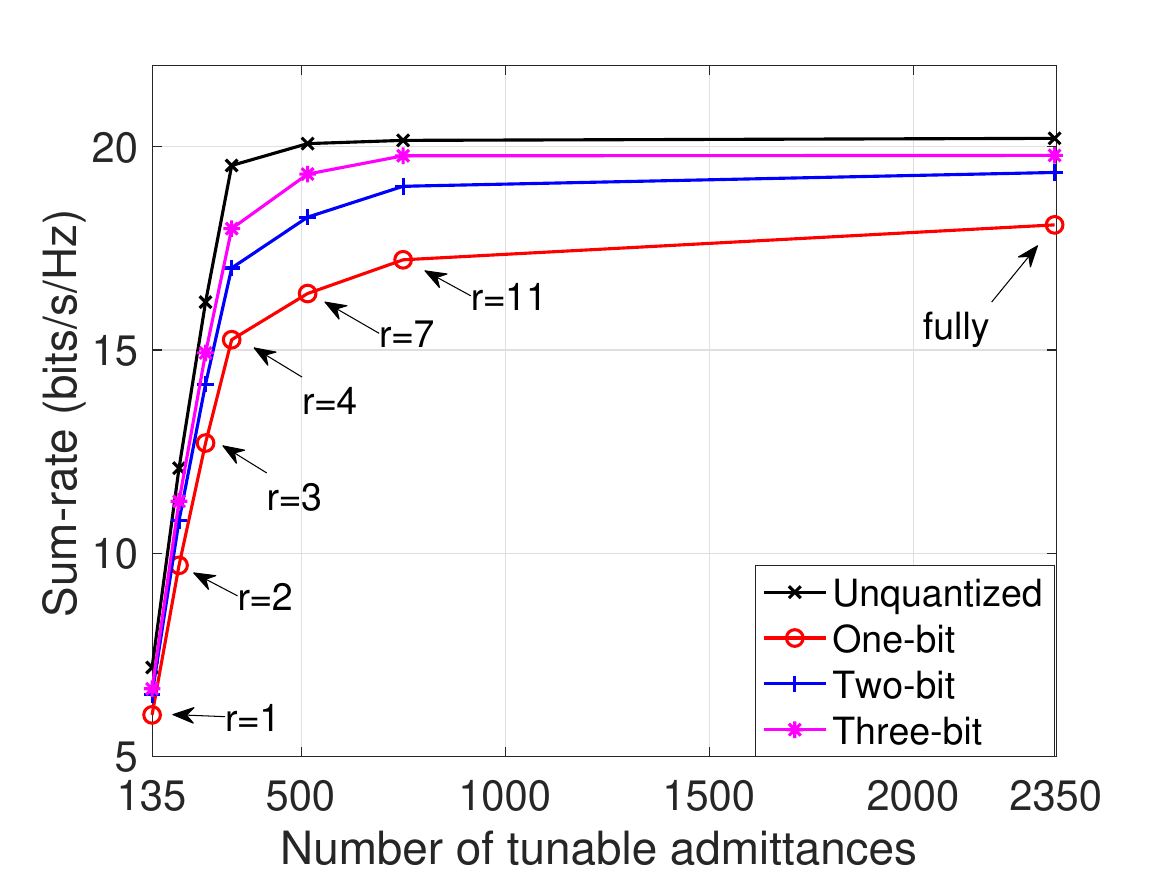}
\centering
\caption{Sum-rate versus the number of tunable admittances for different quantization resolutions, where $N=64$, $K=4$, and SNR $=10$ dB. Different connectivity levels are obtained by varying the number of central ports $r$ in the stem-connected architecture.}
\label{fig:pareto}
\end{figure}
Fig. \ref{fig:pareto} further investigates the tradeoff between hardware complexity and sum-rate performance by varying the number of central ports $r$ in the stem-connected architecture under different quantization resolutions. Several observations can be made. First, in the unquantized case, the stem-connected MiLAC with $r=2K-1=7$ achieves the same performance as the fully-connected MiLAC, consistent with the theoretical result in \cite{zhang2026beamforming}.  Interestingly, the results further show that $r$ can be reduced to $K$ while retaining near-optimal performance. However, reducing $r$ below $K$ causes a significant rate loss because some input ports are not directly connected to the output ports, which severely limits the available input-output transfer paths.  Second, for finite-resolution MiLAC, increasing the resolution from 1 bit to 2 bits yields a significant performance improvement, whereas the additional gain diminishes at higher resolutions. 
Third, the MiLAC upper-bound performance can be approached with  low-resolution quantization and a sparse stem-connected architecture. In particular, a 3-bit stem-connected MiLAC with $r=2K-1$ achieves more than $95\%$ of the sum-rate performance of the unquantized fully-connected case, while using only a 3-bit codebook and $22\%$ of the tunable components required by the fully-connected architecture.   

Overall, the results show that a sparse stem-connected MiLAC architecture with moderate quantization resolution can achieve near-optimal sum-rate performance with significantly lower hardware complexity. As user loading and SNR increase, finer quantization and denser connectivity are required to provide sufficient design flexibility. 


\section{Conclusion}\label{sec:conclusion}
This paper studied finite-resolution MiLAC-aided multiuser beamforming under lossless and reciprocal constraints. For both the online beamforming design and offline codebook design problems, we established exact continuous penalty models that are globally equivalent to their original discrete formulations (when the   penalty parameters are sufficiently large). Based on these models, efficient ADMM-based algorithms were developed for optimizing the finite-resolution MiLAC configuration and the codebook parameter, respectively. Numerical results demonstrated that the proposed online design substantially outperforms the existing projection-based and alternating-refinement methods, while the proposed offline codebook design achieves performance close to exhaustive search. The results also revealed the tradeoff between quantization resolution and circuit connectivity of MiLAC. In particular, a 3-bit stem-connected MiLAC can approach the performance of an unquantized fully-connected MiLAC.  These results suggest that moderate resolution and sparsely connected MiLAC architectures can provide an attractive tradeoff between beamforming performance and hardware complexity.

\appendices
\section{Proof of Theorem \ref{equivalence}}\label{proof:equivalence}
In this appendix, we prove Theorem \ref{equivalence}. We first present two auxiliary results that will be used in the proof in Appendix \ref{subapp:auxiliary}, and then provide the detailed proof of the theorem in Appendix \ref{subapp:proof}. 
\subsection{Auxiliary results}\label{subapp:auxiliary}
To prove Theorem \ref{equivalence}, we first establish a general exactness result for the quadratic penalty reformulation.
\begin{proposition}\label{general:equivalence}
Given $c>0$, consider the discrete optimization problem
\begin{equation*}
\begin{aligned}
 (P_d):  \quad 
    \max_{\mathbf X}\quad & f(\mathbf X) \\
    \mathrm{s.t.}\quad 
    & X_{i,j}\in\{-c,c\},\quad \forall i,j ,
\end{aligned}
\end{equation*}
where $\mathbf X\in\mathbb R^{m_1\times m_2}$ is the optimization variable. Define the penalty function
\begin{equation}
    p(\mathbf X)
    =\sum_{i,j}\left(X_{i,j}^2-c^2\right),
\end{equation}
and consider the continuous penalty problem
\begin{equation*}\label{Pc}
\begin{aligned}
    (P_c):\quad 
    \max_{\mathbf X}\quad & f(\mathbf X)+\rho p(\mathbf X) \\
    \mathrm{s.t.}\quad 
    & X_{i,j}\in[-c,c],\quad \forall i,j .
\end{aligned}
\end{equation*}
Assume that $f$ is $L$-Lipschitz continuous, i.e.,
\begin{equation}
    |f(\mathbf X)-f(\mathbf Y)|\le L\|\mathbf X-\mathbf Y\|_F.
\end{equation}
 Then, for any $
    \rho>\rho_0:=Lc^{-1}$,
any optimal solution $\mathbf X^\star$ of $(P_c)$ satisfies
\begin{equation}
    X_{i,j}^\star\in\{-c,c\},\quad \forall i,j.
\end{equation}
Consequently, problems $(P_c)$ and  $(P_d)$ are equivalent.
\end{proposition}
\begin{proof}
The proof is standard and follows similar steps as in \cite[Theorem 1]{onebit_shao}. For completeness, we provide a detailed proof below. 

Let $\bX^\star$ be an optimal solution of  $(P_c)$.  Define the rounded point $\hat{\mathbf X}$ by
$$    \hat X_{i,j}=c\operatorname{sign}(X_{i,j}^\star).$$ It suffices  to show that $\bX^\star=\hat\bX$. Let $    d_{i,j}=c-|X_{i,j}^\star|\ge 0. $ With this definition, we have
$|\hat X_{i,j}-X_{i,j}^\star|=d_{i,j},$ and hence
\begin{equation}\label{Xhat-Xstar}
  \|\hat{\mathbf X}-\mathbf X^\star\|_F
    =\big(\sum_{i,j}d_{i,j}^2\big)^{1/2}
    \le \sum_{i,j}d_{i,j}.
\end{equation}
By the Lipschitz continuity of $f$,
\begin{equation}\label{functionvalue}   f(\hat{\mathbf X})-f(\mathbf X^\star)
    \ge -L\|\hat{\mathbf X}-\mathbf X^\star\|_F.
\end{equation}
Moreover,
\begin{equation}\label{penalty}
\begin{aligned}
    p(\hat{\mathbf X})-p(\mathbf X^\star) &=
\sum_{i,j}
    \left(c^2-(X_{i,j}^\star)^2\right) \\
    &=\sum_{i,j}
    \left(c-|X_{i,j}^\star|\right)
    \left(c+|X_{i,j}^\star|\right) \\
    &\ge
    {c}\sum_{i,j}d_{i,j},
\end{aligned}
\end{equation}
where the inequality follows from the definition of $d_{i,j}$. 
Combining \eqref{Xhat-Xstar}-\eqref{penalty} gives
\begin{equation}\label{inequality}
\begin{aligned}
    &\left[f(\hat{\mathbf X})+\rho p(\hat{\mathbf X})\right]
    -\left[f(\mathbf X^\star)+\rho p(\mathbf X^\star)\right]
    \\
    &\ge
    -L\|\hat{\mathbf X}-\mathbf X^\star\|_F
    +{\rho c}\sum_{i,j}d_{i,j}
    \\
    &\ge
    \left({\rho c}-L\right)
    \|\hat{\mathbf X}-\mathbf X^\star\|_F,
\end{aligned}
\end{equation}
where the last inequality uses \eqref{Xhat-Xstar}. Therefore, if $\rho>Lc^{-1}$, 
then the right-hand side is strictly positive whenever there exists some $(i,j)$ such that
$
X_{i,j}^\star\neq \hat{X}_{i,j}. 
$
This contradicts the optimality of $\bX^\star$ for $(P_c)$. Hence, ${\bX}^\star=\hat{\bX}$, which completes the proof. 
\end{proof}

The following lemma presents two operations that preserve Lipschitz continuity, which will also be useful in the proof of Theorem \ref{equivalence}. Its proof is standard and omitted for brevity. 
\begin{lemma}[Lipschitz-preserving operations]\label{lemma1}
The following two statements hold.

    (i) \textbf{Composition.} 
    Let $g:\mathbb C^{m_1\times m_2}\to \mathbb C^{n_1\times n_2}$    and $f:\mathbb C^{n_1\times n_2}\to \mathbb C^{q_1\times q_2}$  be Lipschitz continuous with Lipschitz constants $L_g$ and $L_f$, respectively. Then \(f\circ g\) is Lipschitz continuous with constant $L_fL_g$.

    (ii) \textbf{Maximum function.}
    Let
$h(\bX)=\max_{\bY\in\mathcal Y} f(\bX,\bY).$
    Suppose that \(f(\bX,\bY)\) is uniformly Lipschitz continuous with respect to \(\bX\), i.e., there exists \(L>0\) such that
    \[
    |f(\bX_1,\bY)-f(\bX_2,\bY)|\le L\|\bX_1-\bX_2\|_F,~     \forall ~\bX_1,\bX_2, \forall~ \bY\in\mathcal Y.
    \]
    Then \(h(\bX)\) is Lipschitz continuous with constant \(L\).
\end{lemma}

\subsection{Proof of Theorem \ref{equivalence}}\label{subapp:proof}

We now prove Theorem \ref{equivalence}. The main idea is to reduce problem \eqref{problem2:sumrate} to the general form in Proposition \ref{general:equivalence}. To this end, we first eliminate the variables that are uniquely determined by the relaxed two-level variables and then verify the Lipschitz continuity of the resulting objective function.

Let $\bar{\mathbf B}$ and $\bar{\mathbf B}^{(\ell)}$ denote the matrix forms of $\{\bar B_{i,j}\}$ and $\{\bar B_{i,j}^{(\ell)}\}$, respectively, where entries corresponding to $(i,j)\in\mathcal D$ are set to zero. Let $\mathbf q^{(\ell)}$ collect the upper-triangular entries of $\bar{\mathbf B}^{(\ell)}$ indexed by $\mathcal E$, as defined in \eqref{def:ql}, and let $\mathbf q$ denote the stacked vector of $\{\mathbf q^{(\ell)}\}_{\ell=1}^{b}$ as defined in \eqref{def:q}. 

For any given $\mathbf q$, the matrices $\{\bar{\mathbf B}^{(\ell)}\}_{\ell=1}^{b}$ are uniquely determined by the inverse of the vectorization operation in \eqref{def:ql}. Then $\bar{\mathbf B}$ is determined by
\[
\bar{\mathbf B}
=
\sum_{\ell=1}^{b}2^{\ell-1}\bar{\mathbf B}^{(\ell)},
\]
and $\mathbf B$ is further determined by the linear relation in \eqref{con:BbarB2}. Finally, $\bar{\mathbf F}$ is uniquely determined by the bilinear constraint in \eqref{con:BbarF2}, or equivalently,
\begin{equation}\label{proof:barF}
\bar{\mathbf F}
=
(\mathbf I+\mathrm{i}Z_0\mathbf B)^{-1}
(\mathbf I-\mathrm{i}Z_0\mathbf B)\bar{\mathbf I}_1.
\end{equation}
 Hence, we can regard $\bar{\mathbf F}$ as a function of $\mathbf q$, denoted by $\bar{\mathbf F}(\mathbf q)$.

Define
\[
g(\bar{\mathbf F})
=
\max_{\mathbf 1^T\mathbf p\le P_T,\ \mathbf p\ge \mathbf 0}
R(\mathbf p,\bar{\mathbf F};\mathbf H),
\]
and
\[
f(\mathbf q)
=
g(\bar{\mathbf F}(\mathbf q)).
\]
Then problem \eqref{problem2:sumrate} can be  written as 
$$\max_{{\bq}} f(\bq),~\text{s.t. }[\bq]_{i}\in\{-c,c\},~\forall~i.$$
According to Proposition \ref{general:equivalence}, it remains to show that $f(\mathbf q)$ is Lipschitz continuous over the box $\bq\in[-c,c]^{bn_I}$.

We first show that $g(\bar{\mathbf F})$ is Lipschitz continuous. From \eqref{proof:barF},  $\bar{\mathbf F}$ consists of the first $K$ columns of the scattering matrix $\bthe$. Under the lossless and reciprocal model, this scattering matrix is unitary. Hence,
\[
\|\bar{\mathbf F}\|_2\le 1.
\]
Define the compact sets
\[
\mathcal F
=
\{\mathbf X\in\mathbb C^{(N+K)\times K}\mid \|\mathbf X\|_2\le 1\}
\]
and
\[
\Delta
=
\{\mathbf p\in\mathbb R^{K}\mid \mathbf 1^T\mathbf p\le P_T,\ \mathbf p\ge \mathbf 0\}.
\]
For any fixed channel realization $\mathbf H$, the function $R(\mathbf p,\bar{\mathbf F};\mathbf H)$ is continuously differentiable with respect to $(\bar{\mathbf F},\mathbf{p})$ on $\mathcal F\times \Delta$.  Therefore, its gradient with respect to $\bar{\mathbf F}$ is uniformly bounded on the compact set $\mathcal F\times\Delta$, which implies that $R(\mathbf p,\bar{\mathbf F};\mathbf H)$ is uniformly Lipschitz continuous with respect to $\bar{\mathbf F}\in\mathcal{F}$. By Lemma \ref{lemma1}, $g(\bar{\bF})$  is Lipschitz continuous with respect to $\bar{\mathbf F}\in\mathcal{F}$.

Next, we show that $\bar{\mathbf F}(\mathbf q)$ is Lipschitz continuous. It is clear that the mapping from $\mathbf q$ to $\{\bar{\mathbf B}^{(\ell)}\}_{\ell=1}^{b}$, from $\{\bar{\mathbf B}^{(\ell)}\}_{\ell=1}^{b}$ to $\bar{\mathbf B}$, and from $\bar{\mathbf B}$ to $\mathbf B$ are all  linear. Hence, the mapping from $\mathbf q$ to $\mathbf B$ is Lipschitz continuous. It remains to show that the mapping from $\mathbf B$ to $\bar{\mathbf F}$ is Lipschitz continuous. First,  \eqref{proof:barF} can be equivalently expressed as 
$$\bar{\bF}=-\bar{\mathbf{I}}_1+2(\mathbf{I}+\mathrm{i}Z_0\bB)^{-1}\bar{\mathbf{I}}_1.$$
 Given $\bB_1, \bB_2\in\R^{(N+K)\times (N+K)}$, let 
$$\bar{\mathbf F}_1=-\bar{\mathbf{I}}_1+2(\mathbf{I}+\mathrm{i}Z_0\bB_1)^{-1}\bar{\mathbf{I}}_1$$ 
and
$$\bar{\mathbf F}_2=-\bar{\mathbf{I}}_1+2(\mathbf{I}+\mathrm{i}Z_0\bB_2)^{-1}\bar{\mathbf{I}}_1.$$
We have 
$$
\begin{aligned}
&\|\bar{\bF}_1-\bar{\bF}_2\|_F\\
&\overset{(a)}{\leq} 2\|(\mathbf{I}+\mathrm{i}Z_0\bB_1)^{-1}-(\mathbf{I}+\mathrm{i}Z_0\bB_2)^{-1}\|_F\\
&\overset{(b)}{=} 2\|\mathrm{i}Z_0(\mathbf{I}+\mathrm{i}Z_0\bB_1)^{-1}(\bB_2-\bB_1)(\mathbf{I}+\mathrm{i}Z_0\bB_2)^{-1}\|_F\\
&\leq 2Z_0\|(\mathbf{I}+\mathrm{i}Z_0\bB_1)^{-1}\|_2\|(\mathbf{I}+\mathrm{i}Z_0\bB_2)^{-1}\|_2\|\bB_2-\bB_1\|_F\\
&\overset{(c)}{\leq} 2Z_0\|\bB_2-\bB_1\|_F.
\end{aligned}$$
where (a) uses the inequalities $\|\bX\bY\|_F\leq \|\bX\|_2\|\bY\|_F$ and  $\|\bar{\mathbf{I}}_1\|_2\leq 1$, (b)  applies $\bX^{-1}-\bY^{-1}=\bX^{-1}(\bY-\bX)\bY^{-1}$, and (c) holds since for any real symmetric matrix $\bB$,  $\|(\mathbf{I}+\mathrm{i}\bB)^{-1}\|_2\leq 1$. Hence, the mapping from $\bB$  to $\bar{\bF}$ is Lipschitz continuous in ${\bB}$. 
Combining this fact with the Lipschitz continuity of the mapping from $\mathbf q$ to $\mathbf B$, we conclude that $\bar{\mathbf F}(\mathbf q)$ is Lipschitz continuous with respect to $\mathbf q$.

Since $g(\bar{\mathbf F})$ and $\bar{\mathbf F}(\mathbf q)$ are both Lipschitz continuous, Lemma \ref{lemma1} implies that $f(\mathbf q)=g(\bar{\mathbf F}(\mathbf q))$  is Lipschitz continuous. This completes the proof. 

\section{Proof of Theorem \ref{equivalence2}}\label{app:equivalence2}
This appendix proves Theorem \ref{equivalence2}. Appendix \ref{subapp:general} first extends Proposition \ref{general:equivalence} to the case where $c$ is also an optimization variable, and Appendix \ref{subapp:equivalence} then applies this result to prove Theorem \ref{equivalence2}.
\subsection{A General Result}\label{subapp:general}
\begin{proposition}\label{general:equivalence2}
Consider the discrete optimization problem
\begin{equation*}
\begin{aligned}
 (P_d):  \quad 
    \max_{\mathbf X,c\geq 0}\quad & f(\mathbf X) \\
    \mathrm{s.t.}~\quad 
    & X_{i,j}\in\{-c,c\},\quad \forall i,j ,
\end{aligned}
\end{equation*}
where $\mathbf X\in\mathbb R^{m_1\times m_2}$ and $c\ge 0$ are both optimization variables. Define the penalty function
\begin{equation}
    p(\mathbf X,c)
    =\sum_{i,j}\left(X_{i,j}^2-c^2\right),
\end{equation}
and consider the continuous penalty problem
\begin{equation*}\label{Pc}
\begin{aligned}
    (P_c):\quad 
    \max_{\mathbf X,c\geq 0}\quad & f(\mathbf X)+\rho p(\mathbf X,c) \\
    \mathrm{s.t.}~\quad 
    & X_{i,j}\in[-c,c],\quad \forall i,j .
\end{aligned}
\end{equation*}
Assume that $f$ is $L$-Lipschitz continuous and $f(\mathbf 0)=0$. Suppose that there exists a feasible point $(\bar{\mathbf X},\bar c)$ of $(P_d)$ such that
\begin{equation}
    f(\bar{\mathbf X})=\bar v>0.
\end{equation}
Then, for any $
    \rho>\bar{\rho}\triangleq {(m_1m_2)^{1/2}L^2}{\bar v}^{-1}$,
any optimal solution $(\mathbf X^\star,c^\star)$ of $(P_c)$ satisfies
\begin{equation}
    X_{i,j}^\star\in\{-c^\star,c^\star\},\quad \forall i,j.
\end{equation}
\end{proposition}

\begin{proof}
Let $(\mathbf X^\star,c^\star)$ be an optimal solution of $(P_c)$.  We first show that $c^\star$ is lower bounded. Let $n=m_1m_2$. For any feasible point $(\mathbf X,c)$ of $(P_c)$, we have
 $p(\mathbf X,c)\le 0$ and 
$\|\mathbf X\|_F\le \sqrt n c. $
Using $f(\mathbf 0)=0$ and the Lipschitz continuity of $f$, we further obtain
\begin{equation}
  f(\mathbf X)+\rho p(\mathbf X,c)\le   f(\mathbf X)\le L\|\mathbf X\|_F\le L\sqrt n c.
\end{equation}
Since $(\bar{\mathbf X},\bar c)$ is feasible for $(P_d)$, it is also feasible for $(P_c)$. In addition, because
$\bar X_{i,j}\in\{-\bar c,\bar c\}$, we have
$    p(\bar{\mathbf X},\bar c)=0.
$ Hence, the optimal value of $(P_c)$ is at least
\begin{equation}
    f(\bar{\mathbf X})+\rho p(\bar{\mathbf X},\bar c)
    =\bar v>0.
\end{equation}
Then
\begin{equation}
    \bar v
    \le f(\mathbf X^\star)+\rho p(\mathbf X^\star,c^\star)
    \le L\sqrt n c^\star.
\end{equation}
Thus,
\begin{equation}
    c^\star\ge \underline c
    \triangleq \frac{\bar v}{L\sqrt n}>0.
\end{equation}

Next, we prove that $X_{i,j}^\star\in\{-c^\star,c^\star\}$.  Define the rounded point $\hat{\mathbf X}$ by
$$    \hat X_{i,j}=c^{\star}\operatorname{sign}(X_{i,j}^\star).$$  Following the same procedures as in \eqref{Xhat-Xstar}-\eqref{inequality}, we get 
$$
\begin{aligned}
    &\left[f(\hat{\mathbf X})+\rho p(\hat{\mathbf X},c^\star)\right]
    -\left[f(\mathbf X^\star)+\rho p(\mathbf X^\star,c^\star)\right]
    \\
        &\ge
    \left({\rho\underline c}-L\right)
    \|\hat{\mathbf X}-\mathbf X^\star\|_F.
\end{aligned}
$$
 Therefore, if 
  $  \rho>{L}{\underline c}^{-1}
    =
    {n^{1/2}L^2}{\bar v}^{-1},
$  we can conclude that ${\bX}^\star=\hat{\bX}$. 
\end{proof}
\subsection{Proof of Theorem \ref{equivalence2}}\label{subapp:equivalence}
As in the proof of Theorem \ref{equivalence}, for each channel sample $m$, let $\mathbf q^{(m)}\in\mathbb R^{bn_I\times 1}$ denote the vector that stacks all independent upper-triangular two-level variables $\bar{B}_{i,j}^{(m,\ell)}$ with $(i,j)\in\mathcal E$ and $j\ge i$.  Then, problem \eqref{problem:codebook2} can be equivalently expressed as an optimization problem in terms of $\{\mathbf q^{(m)}\}_{m=1}^{M}$. Following the same steps as in Appendix \ref{subapp:proof}, we can show that the resulting objective function is Lipschitz continuous in $\{\mathbf q^{(m)}\}_{m=1}^{M}$. The details are omitted to avoid repetition.

 It remains to verify the following two conditions: (i) When $c=0$,  the objective value is zero; and (ii) there exists $c>0$ and a discrete feasible point satisfying $\mathbf q^{(m)}\in\{-c,c\}^{bn_I}, m=1,2,\ldots,M$, such that the corresponding objective value is strictly positive.
 
  Condition (i) is straightforward. When $c=0$, we have $\bB^{(m)}=\mathbf{0}$ for all $m=1,2,\dots, M$. It follows that $\bar{\bF}^{(m)}=\bar{\mathbf{I}}_1$ and $\bar{\mathbf{I}}_2\bar{\bF}^{(m)}=\mathbf{0}$. Hence, the average sum-rate objective function of problem \eqref{problem:codebook2} is zero.     We next show that  Condition (ii) is satisfied if there exists at least one input port $k_0\in\{1,2,\ldots,K\}$ connecting to an output port $r_0\in\{K+1,K+2,\ldots,N+K\}$, i.e., $(k_0,r_0)\in\mathcal{E}$, as assumed in Section \ref{subsubsec:arch}. 
For all \(m=1,2,\ldots,M\), set
\[
\bar B_{i,j}^{(m,\ell)}=c,\quad \forall~\ell=1,\ldots,b,~\forall~(i,j)\in\mathcal{E}.
\]
It follows that
\begin{equation}\label{allocate:Bm}
\bar B_{i,j}^{(m)}=\sum_{\ell=1}^{b}2^{\ell-1}\bar B_{i,j}^{(m,\ell)}=(2^b-1)c,~\forall~(i,j)\in\mathcal{E}.
\end{equation}
 Define 
\(
\mathbf B_0^{(m)}=\frac{1}{c}\mathbf B^{(m)},
\)
where \(\mathbf B_0^{(m)}\) is independent of \(c\). Then
\[
\bar{\bF}^{(m)}=2\big(\mathbf I+\mathrm{i}Z_0c\mathbf B_0^{(m)}\big)^{-1}\bar{\mathbf{I}}_1-\bar{\mathbf{I}}_1.
\]
For sufficiently small \(c>0\), the Neumann-series expansion gives
\[\bar{\bF}^{(m)}=\bar{\mathbf{I}}_1-2\mathrm{i}Z_0c\mathbf B_0^{(m)}\bar{\mathbf{I}}_1+{\cal O}(c^2).\]
According to \eqref{allocate:Bm}, 
$$[{\bB}^{(m)}]_{r_0,k_0}=-\bar{B}_{r_0,k_0}^{(m)}=-(2^b-1)c,$$
and hence 
$$
\begin{aligned}
\bigl[\bar{\bF}^{(m)}\bigr]_{r_0,k_0}&=2\mathrm{i}Z_0c(2^b-1)+{\cal O}(c^2),\\
\end{aligned}
$$
which is nonzero for sufficiently small $c>0$. Hence, $|(\bh_{k_0}^{(m)})^H\bar{\mathbf{I}}_2\bar{\bF}^{(m)}_{:,k_0}|^2>0$ with probability one. 
 This further implies that the average sum-rate objective is strictly positive with probability one by allocating all transmit power to the $k_0$-th user,  which verifies Condition (ii).

To conclude, the problem in \eqref{problem:codebook2} satisfies all the conditions in Proposition \ref{general:equivalence2}. Hence, the proof is complete.

 \bibliographystyle{IEEEtran}
 \bibliography{IEEEabrv,milac_discrete}

\end{document}